\documentclass[12pt]{article}

\usepackage{fullpage}
\usepackage[T1]{fontenc}
\usepackage{ae,aecompl}
\usepackage[dvips]{graphicx}
\usepackage{amssymb}
\usepackage{amsmath}
\usepackage[a4paper, margin=1.25in]{geometry}
\usepackage[round,authoryear]{natbib}
\usepackage{color}
\usepackage{array}
\definecolor{webgreen}{rgb}{0,0.4,0}
\definecolor{webbrown}{rgb}{0.6,0,0}
\definecolor{purple}{rgb}{0.5,0,0.25}
\definecolor{darkblue}{rgb}{0,0,0.7}
\definecolor{darkred}{rgb}{0.7,0,0}
\definecolor{darkgreen}{rgb}{0,0.7,0}
\usepackage[pdfborder=false]{hyperref}
\hypersetup{colorlinks,citecolor=darkred,filecolor=black,linkcolor=darkblue,urlcolor=webgreen,pdfpagemode=none,
pdfstartview=FitH}
\newcommand{\ignore}[1]{}
\newtheorem{lemma}{{\sc Lemma}}[section]
\newtheorem{prop}{{\sc Proposition}}[section]
\newtheorem{cor}{{\sc Corollary}}[section]
\newtheorem{theorem}{{\sc Theorem}}[section]
\newtheorem{defn}{{\sc Definition}}[section]

\newtheorem{claim}{{\sc Claim}}
\newtheorem{example}{{\sc Example}}[section]

\newtheorem{remark}{{\sc Remark}}

\newenvironment{proof}{\noindent {\bf \sl Proof\/}:\enspace}
{\hfill $\blacksquare{}$ \vspace{12pt}}
\usepackage{sectsty}
\sectionfont{\centering\normalfont\scshape}
\subsectionfont{\centering\normalfont}

\begin{document}

\title{\bf Optimal Strategy-proof  Mechanisms on Single-crossing Domains}
\author{\bf Mridu Prabal Goswami\thanks{Indian Statistical Institute, Tezpur, India. 
		 I am extremely grateful to Dipjyoti Majumdar, Debasis Mishra, Arup Pal, Abinash Panda, Soumendu Sarkar, Arunava Sen and Tridib Sharma for many helpful comments and suggestions.   
		I am thankful to a research grant from SERB, DST, Government of India, and Magesh Kumar K. K. and Manish Yadav.}}
\maketitle

\begin{abstract} \noindent We consider an economic environment with one buyer and one seller. For a bundle $(t,q)\in [0,\infty[\times [0,1]=\mathbb{Z}$, $q$ refers to the winning probability of an object, and $t$ denotes the payment that the buyer makes. We consider continuous and monotone preferences on $\mathbb{Z}$ as the primitives of the buyer. These preferences can incorporate both quasilinear and non-quasilinear preferences, and multidimensional pay-off relevant parameters. We define rich single-crossing subsets of this class and characterize strategy-proof mechanisms by using monotonicity of the mechanisms and continuity of the indirect preference correspondences. We also provide a computationally tractable optimization program to compute the optimal mechanism for mechanisms with finite range. We do not use revenue equivalence and virtual valuations as tools in our proofs. Our proof techniques bring out the geometric interaction between the single-crossing property and the positions of bundles $(t,q)$s in the space $\mathbb{Z}$. We also provide an extension of our analysis to an $n-$buyer environment, and to the situation where $q$ is a qualitative variable. 
\end{abstract}


\section{\bf Introduction}
We consider an economic environment with one buyer and one seller. The consumption space of the buyer is $[0,\infty[\times [0,1]=\mathbb{Z}$. We call an ordered pair $(t,q)\in \mathbb{Z}$ a bundle. 
Here $q$ refers to the winning probability of an object, and $t$ denotes the payment that the buyer makes in return.
We consider preferences defined on $\mathbb{Z}$. These preferences are monotone in both $t$ and $q$, and are continuous. Further, we consider domains that satisfy the single-crossing property, we call such domains single-crossing. A mechanism is a function from a single-crossing domain to $\mathbb{Z}$. A domain satisfies the single-crossing property if indifference sets of two distinct preferences that belong to the domain intersect at most at one bundle. 
We assume that the buyer knows, and only the buyer knows her preference, i.e., preference is private information of the buyer.
The seller wants to elicit this private information.     
Hence, we study mechanisms that are strategy-proof, i.e., mechanisms where the buyer has no incentives to misreport her preference. 

A very important class of single-crossing domains for which strategy-proof mechanisms are well studied is 
the domain of quasilinear preferences, i.e., linear in payment. An example of such a domain is given by the preferences $\theta q-t, \theta \in [1,2]$.
\citep{myer} and  \citep{Hag and Rog} are two classic papers that consider such domains. We observe that the single-crossing property provides a general structure on domains of preferences. For example, $\theta q-t^2, \theta \in [1,2]$ which is a single-crossing domain of preferences that are not quasilinear. According to a notion in   
\citep{Mishra1} the preferences from this domain satisfy the positive income effect.  
Further, single-crossing domains can also incorporate domains of preferences that admit multidimensional parametric representations, we give Example \ref{example:two-parameter} to illustrate this observation. The domain in Example  \ref{example:two-parameter} has a behavioral explanation. In Example \ref{example:two-parameter} if the valuation for the good is not high enough, then the loss of utility from  paying for the good is high. If the valuation for the good meets a cutoff, then the loss of utility from making a payment is low. 
From these examples we observe that the single-crossing is a common unifying property of many important classes
of preferences where each class represents a distinct behavioral trait. Since the single-crossing property is defined by intersection of indifference sets of two distinct preferences, it is an ordinal feature of a domain, i.e., this property does not depend on utility representations of preferences.          
We utilize this ordinal property to study strategy-proof mechanisms. In particular,        
the single-crossing property entails an order on the set of preferences. The order on the domain ensures that we can define monotonicity of the mechanisms. By imposing a natural richness condition on the single-crossing domains we can define a natural order topology on single-crossing domains. This order topology enables us to define continuous indirect preference correspondences. 
The topology in \citep{myer} is the standard Euclidean metric topology.  
That is, in \citep{myer} the ordered space is a subset of the real line with the Euclidean topology added. The Euclidean topology is also the order topology generated by the open intervals of the real line. This topological structure makes it easier to use standard calculus to study the auction design problem in \citep{myer}. The ordered space of the single-crossing preferences is a linear continuum like the real line, also the topology generated by the collection of open intervals of this ordered space is metrizable. However, we  do not assume any smooth structure on it, unlike Euclidean metric in \citep{myer}. With the help of our axioms we show that if the ranges of mechanisms are finite, then just by using the ordinal properties, i.e., without any assumption of smoothness,  of the preferences we can characterize strategy-proof mechanisms. 
In particular, for mechanisms with finite range the single-crossing property entails a geometric structure of  strategy-proof mechanisms. This geometry can be used to compute the expected revenue from a strategy-proof mechanism easily. This geometry can be written as an ordered tuple in the Cartesian products of the spaces of bundles and preferences, and the set of all geometry is compact in the product topology, which provides a simple existence proof of optimal mechanisms. 
Further, we do not require any topological assumption on the space of mechanisms to ascertain the existence of an optimal mechanism.   
In \citep{Goswami2} we show that if the expected revenue of a mechanism whose range is not finite, and its revenue is not smaller than the expected revenue from a finite range mechanism, then the revenue of that mechanism can be approximated by revenues from a sequence mechanisms with finite range.      
We provide an extension of the $1-$buyer environment to a $n-$buyer environment that shows how the procedure to compute the optimal mechanism in the $1-$buyer environment can be used in a straightforward way in the $n-$buyer environment.

We show that the monotonicity of a mechanism, and the continuity of the corresponding indirect preference correspondence are necessary for strategy-proofness. In Theorem \ref{thm:implies_strtagey_proof_finite_range}  we show that if the range of a mechanism is finite, monotone and the indirect preference correspondence is continuous, then the mechanism is strategy-proof. The proof of Theorem \ref{thm:implies_strtagey_proof_finite_range} is constructive.
Given the range, i.e., given the positions of the bundles in the space $\mathbb{Z}$, the proof of the theorem provides the exact rule of the mechanism. Since our proofs depend only on the ordinal feature of the single-crossing property we do not use revenue equivalence, and in general envelope theorems, in our characterization. A standard approach to study strategy-proof mechanism 
is to first consider a monotone allocation rule, and then find a payment rule so that the given allocation rule and the payment rule that is found together define a strategy-proof mechanism. Our approach is different. 
We consider the mechanism as a whole, i.e., we consider the allocation and the payment rule together, and endow axioms on the mechanisms and the corresponding indirect preference correspondences and study strategy-proofness. We have extended this idea to study strategy-proof mechanisms on a class of non separable preferences, for example preferences that incorporate risk aversions, in \citep{Goswami1}.        
Although we consider a general topological space of preferences, our proofs are simple.
The paper is organized as follows. We prove the existence of optimal mechanism with finite range with a maximum cardinality in Theorem \ref{thm:optimal}. In Section \ref{sec:pre} we introduce important definitions.  
In Section \ref{sec:sp} we study strategy-proof mechanisms. 
In Section \ref{sec:opt} we study optimal mechanisms. In Section \ref{sec:quality} we make remarks about how we can extend our analysis when $q$ is a qualitative variable like quality or share.    
In Section \ref{sec:union_single_crossing} we make some observations about extending a domain beyond rich single-crossing domains.  
In Section \ref{sec:n_buyer} we provide an extension of our analysis to an $n-$buyer environment. In Section \ref{sec:con} we make some concluding remarks.  We discuss some related papers next.

\subsection{\bf Related Literature}
\label{sec:lit}  
\citep{Baisa2} considers a specific one parameter class of preferences that satisfy the single-crossing property. We allow for multidimensional parametric representations of preferences. \citep{Saporiti} considers a single-crossing domain to study  voting rules. Instead of assuming an order on the preferences which is the case in \citep{Baisa2}; in \citep{Baisa2} the order is due to the natural order on the parameters that represent the preferences, \citep{Tian} and \citep{Saporiti} we derive an order  on the set of preferences by using monotonicity and continuity  of the preferences.
The single-crossing condition in \citep{Laffont1} use differentiable utility functions. 
In our definition we do not use any parametric class of utility functions. 
\citep{Laffont1} consider a one dimensional parametric class of single-crossing preferences and study 
the implementability of piecewise continuously differentiable allocation rules.
We do  not assume any smooth structure on preferences.
 \citep{Tian} study equivalence between a notion of implementability and notions of cycle monotonicity. 
The implementability condition is defined by the operator ``$\max$''.  Our approach is ordinal. Our axioms about monotonicity of mechanisms and continuity of the indirect preference correspondences are ordinal. 
\citep{Gershkov} construct constrained
efficient optimal mechanism on single-crossing domains as in \citep{Saporiti}.
Furthermore  in the context of voting \citep{Barbera2} study a model where society’s preferences over
voting rules satisfy the single-crossing property with an objective to analyze self-stable rather than strategy-proof voting rules. \citep{Gans} study
an Arrovian aggregation problem with single-crossing preferences for voters, and 
show that median voters are decisive in all majority elections between pairs of 
alternatives. \citep{Barbera3} develop a notion of `top monotonicity' which is a common generalization of single-peakedness and single-crossingness. 
\citep{Corchon} study a public-good-private-good production economy where
agents' preferences satisfy the single-crossing property and prove that smooth strategy-proof and Pareto-efficient social choice functions that give strictly positive amount
of both goods to all agents do not exist. 
Since the single-crossing property allows for various kinds of non linearity in the preferences, the literature on non-quasilinear preferences is also relevant.  Some of the recent studies on non-quasilinear preferences are  
\citep{Mishra1} and  \citep{Mishra2} \citep{Serizawa}, \citep{Baisa1}. 
The main idea of the single-crossing domain in this paper is from \citep{Goswami}. In \citep{Goswami} the single-crossing property is defined for classical exchange economies. Uses of single-crossing conditions in contract design goes back to \citep{Spence}, \citep{Mirr}       
and \citep{Roths}. An use of the order on the set of preferences ensued from the single-crossing property to study optimal mechanism more generally is a novel aspect of this paper.

\section{\bf Preliminaries}
\label{sec:pre}
The economic environment in this paper consists of a seller and a buyer. 
The seller sells an indivisible unit of a good. Let $q\in [0,1]$ denote the probability that the object is sold to the buyer.    
In return, the buyer needs to make a payment to the seller. This  payment is denoted by $t$. The set of allocations is denoted by 
$\mathbb{Z}$, and $\mathbb{Z}=[0,\infty[\times [0,1]$, where $[0,\infty[=\Re_{+}$ denotes the set of non-negative real numbers.
A typical bundle is denoted by 
$(t,q)$, where $t\in \Re_{+}$ and $q\in [0,1]$.  The buyer's preference over $\mathbb{Z}$ is denoted by $R$. 
The strict counterpart of $R$ is denoted by $P$, and indifference is denoted by $I$.
 For $z\in \mathbb{Z}$, and $R$, let $UC(R,z)=\{z'\in \mathbb{Z}| z' Rz\}$. In words $UC(R,z)$ is the set of bundles that are weakly preferred to $z$ under $R$.
Likewise $LC(R,z)=\{z'|zRz'\}$, $LC(R,z)$ is the set of bundles that are weakly less preferred to $z$ under $R$. 
We assume the preferences to be continuous and monotone, in short we call the CM preferences.       
The notion of a CM preference is defined formally below.

\begin{defn}[ CM Preference] \rm 
	The complete, transitive preference relation $R$ on $\mathbb{Z}$ is \textbf{CM} if $R$ is  
	{\bf monotone}, i.e., 
	
	\begin{itemize}
		
		\item \textbf{money-monotone:} for all $q\in [0,1]$, if $t''>t'$, then $(t',q)P(t'',q)$.
		\item \textbf{$q$-monotone:} for all $t\in \Re_+$, if $q''>q'$, then $(t,q'')P(t,q')$.

	\end{itemize}
	
	\noindent and  \textbf{continuous} for each $z\in \mathbb{Z}$, the sets $UC(R,z)$ and $LC(R,z)$ are closed sets.{\footnote{These two sets are closed in the product topology on the Euclidean space $\mathbb{Z}$.    }}
	
	\label{defn:CM} 
\end{defn}

\noindent \citep{Mishra1} call a CM preference classical.  
Let $R$ be a CM preference and $x$ a bundle, define  $IC(R,x)=\{y\in \mathbb{Z}\mid yIx\}$.{\footnote{An $IC(R,x)$ set also represents an equivalence class of the equivalence relation $I$.}} The set $IC(R,x)$ is the set of bundles that are indifferent to $x$ according to the preference $R$. 
It can be seen easily that due to the properties of a CM preference, an $IC$ set can be represented as a curve in $\mathbb{Z}$. Thus, we may also call an $IC$ set an $IC$ curve. We shall represent $\Re_{+}$ on the horizontal axis, and $[0,1]$ on the vertical axis. Next we make a remark about considering CM preferences to be primitives of our approach.

\begin{remark}\rm The primitive of our approach is a CM preference of the buyer over her probability of win and payment, and not preferences over lotteries on end wealth positions. Let us consider two CM preferences, $u(t,q;\theta)=\theta q-t$ and $v(t,q;\theta)=\theta q -t^2$ where $\theta>0$. The preference $u$ can be given an expected utility interpretation. Consider a lottery that entails end wealth position $\theta-t$ with probability $q$ and $-t$ with probability $1-q$. Then $u(t,q;\theta)=q(\theta-t)+(1-q)(-t)$. Now $v(t,q;\theta)=q(\theta-t^2)+(1-q)(-(t)^2)$. Although, $v$ is also an expected utility with consequences $(\theta-t^2)$ and $-(t)^2$, these are not the end wealth positions. Thus $v$ is not an expected utility over a lottery over end wealth positions.   
\end{remark}

\noindent Next we make a remark about the payment variable in our model. 

\begin{remark}\rm  Consider a model where the buyer pays only when she pays. 
Let the buyer pays $m$ if she wins an auction, and pays $0$ if she does not win.
Then expected payment is $t=mq$. Thus a quasilinear utility function
$\theta q-t$ is a CM preference since $mq$ is just non negative real number. However, if the buyer's preference is not defined over $(t,q)$s and instead over $(m,q)$, then her preference is $\theta q-qm$ which is not a CM preference over $(m,q)$s. The preference $\theta q-qm$ is not monotone and in fact it is a non-separable preference. We study this preference in \citep{Goswami1}.      
\end{remark}	 	    

\noindent An $IC$ curve in $\mathbb{Z}$ is an upward slopping curve, i.e., if $(t',q'),(t'',q'')\in IC(R,x)$ and $t'<t''$, then $q'<q''$. Let $x'=(t',q'), x''=(t'',q'')$. By $x'\leq x''$ we mean either $x'=x''$ or $t'<t'',q'<q''$. Further, by $x'<x''$ we mean $t'<t'', q'<q''$.     
We call two bundles $x'=(t',q'), x''=(t'',q'')$ {\bf diagonal} if $x'< x''$. 
The single-crossing property is defined next.

\begin{defn}[Single-Crossing of two Preferences]\rm
	We say that two distinct CM preferences $R', R''$ exhibit the {\bf single-crossing property} if and only if  for all $x,y,z \in \mathbb{Z}$, 
	$$\text{if}~ z\in IC(R',x)\cap IC(R'',y),~\text{then}~ IC(R',x)\cap IC(R'',y)=\{z\}.$$
	\label{defn:single_crossing}
\end{defn}
\noindent The single-crossing property implies that two $IC$ curves of two distinct preferences can meet ( or cut) at most one bundle. The single-crossing property is an ordinal property of preferences, i.e., this property does not depend on utility representations of preferences.
Next we define the notion of a single-crossing domain.{\footnote{The monotonicity of preferences and the single-crossing property in \citep{Goswami} are defined for the interior of the consumption space. }}            

\begin{defn}[Rich Single-crossing domain]\rm We call a subset of the set of CM preferences {\bf single-crossing domain} if any $R',R''$ that belongs to the subset 
	satisfy the single-crossing property. We call a single crossing domain {\bf rich} if for any two bundles $x'=(t',q'), x''=(t'',q'')$ such that $t'<t'',q'<q''$ there is $R$ in the single crossing domain such that 
	$x'Ix''$. We denote a rich single-crossing domain by $\mathcal{R}^{cms}$.        	
\end{defn}

\noindent We may interpret the single-crossing property as an algorithm that produces  restricted domains. Consider the set of all CM preferences. Consider two CM preferences $R'$ and $R''$ that admit the single-crossing property. Suppose we wish to add another preference $R'''$ to the collection that already contains $R'$ and $R''$. 
The single-crossing property ensures that 
an additional preference can be added to the collection only in a specific manner. 
A rich single-crossing domain is a maximal single-crossing domain, i.e., $\mathcal{R}^{cms}\cup \{R\}$, where $R\not\in \mathcal{R}^{cms}$ and $R$ is CM, is not a single-crossing domain. For all diagonal bundles $x', x''$ there is a preference $R\in \mathcal{R}^{cms}$ such that $x'\in IC(R,x'')$. Thus if another preference is added to $\mathcal{R}^{cms}$, then it violates the single-crossing property. 
Given the initial preferences  $R'$ and $ R''$, a rich single-crossing domain may be interpreted as the limit point of the algorithm. If we add preferences to a rich single-crossing domain, then we may have a situation pertaining to Maskin Monotonic Transformations which is defined next.
In general, by adding a preference to $\mathcal{R}^{cms}$ we allow indifference curves of two preferences 
to be tangential.  

\begin{defn}\rm ({\bf Maskin Monotonic Transformation}) Let $R',R''$ be two CM preferences. We say 
	that $R''$ is a {\bf Maskin Monotonic Transformation, (in short MMT)} of $R'$ through $z$, if 
	$(i)$ $UC(R'',z)\subseteq UC(R',z)$, $(ii)$ if $x\neq z, x\in UC(R'',z)$, then $xP'z$.  	
\end{defn}

\noindent MMT implies that the indifference curve of $R''$ through $z$ is tangential to the indifference curve of $R'$ through $z$. Single-crossing domains do not allow MMTs in the interior of $\mathbb{Z}$. 
In Section \ref{sec:union_single_crossing} we see that adding preferences to a CMS domain by allowing for MMT we add behavioral traits to a CMS domain.    
The set of all CM preferences satisfy MMTs for all preferences at every bundle. In the context of two-good and two-agent exchange economies \citep{Barbera1} find that a strategy-proof and individually rational mechanism must have a range whose elements fall on at most two line segments. \citep{Goswami} provides an example to demonstrate this not to be the case if the domain is rich single-crossing. Thus, the geometry of the range of strategy-proof mechanisms change if the domains of mechanisms are larger than rich single-crossing domains. 
We next make some brief remarks about the topology on $\mathcal{R}^{cms}$, details of which can be found in Appendix $1$.

The single-crossing property provides a natural way to define an order on $\mathcal{R}^{cms}$. We denote this order by $\prec$.  
This order defines an order topology on $\mathcal{R}^{cms}$. We use this order topology to formulate our axioms. -crossing maybe satisfied without the preferences being CM, \citep{Goswami} provides a example to this effect. Further, we use the notion of infimum and supremum in our proofs. For infimum and supremum to be well defined we need an ordered space to satisfy the least upper bound property. The ordered space that we study, i.e., space of a rich CM single-crossing preferences, satisfy the least upper bound property. Further, this order topology is metrizable, more specifically with this topology $\mathcal{R}^{cms}$ is homeomorphic to the real line with the standard Euclidean metric. This makes this space an one dimensional topological manifold. We do not assume any further structure on $\mathcal{R}^{cms}$, in particular we do not assume any smooth structure on this space. To follow the rest of the paper the definition of the order is required. Thus we state the order next. Let $\square(z)=\{x\mid x\leq z\}$.

\begin{defn}\rm 
	Let $\mathcal{R}^{cms}$ be a rich single crossing domain. Consider $z\in \mathbb{Z}$ and $R', R''\in \mathcal{R}^{cms}$. We say, $R''$ \textbf{cuts} $R'$ \textbf{from above} at $z\in \mathbb{Z}$, if and only if  
	\[ \square(z)\cap UC(R'',z) \subseteq \square(z)\cap UC(R',z). \]
	\noindent We say that $R''$ cuts $R'$ from above if $R''$ cuts $R'$ from above at every bundle. In this case we say $R'\prec R''$. 
	\label{defn:cut} 
\end{defn}

\noindent In words, the indifference curve of $R''$ through $z$ lies above the indifference curve for $R'$ through $z$ in $\square(z)$ if both indifference curves are viewed from the horizontal axis $\Re_{+}$. It is argued in Appendix $1$ that if $R''$ cuts $R'$ from above at some $z\in \mathbb{Z}$, then $R''$ cuts $R'$ from above at every $z\in \mathbb{Z}$ which makes the order well defined. Properties of CM preferences are used to establish this fact.    
In this topology $\mathcal{R}^{cms}$ is connected and closed interval $[R',R'']=\{R\mid R'\precsim R \precsim R''\}$ is compact, here $R'\precsim R$ means either $R'\prec R$ or $R'=R$, $]-\infty, R[$ and $]-\infty, R]$ denote open and closed intervals respectively that are not bounded below, $]R,\infty[$ and $[R,\infty]$ denote open and closed intervals respectively that not bounded above.
All ordered sets  in this paper are written by using the braces $]$ and $[$. For example, if $\alpha,\beta \in \Re $ and $\alpha<\beta$, then $]\alpha,\beta]$ denotes the left open and right closed interval.     
Convergence of a sequence in $\mathcal{R}^{cms}$ is equivalent to saying that monotone sequences converge. A sequence of preferences is denoted by 
$\{R^{n}\}_{n=1}^{\infty}$. A monotone decreasing sequence means any $n$, $R^{n+1}\precsim R^{n}$, and its convergence to $R$ is denoted by $R^{n}\downarrow R$. A monotone 
increasing sequence means for any $n$, $R^{n}\precsim R^{n+1}$, and its convergence to $R$ is denoted by $R^{n}\uparrow R$. In general convergence to $R$ is denoted by $R^{n}\rightarrow R$.
In the next subsection we provide examples of rich single-crossing domains.

\subsubsection{\bf Examples of $\mathcal{R}^{cms}$} 	

\noindent The first example is of quasilinear preferences.    

\begin{example}\rm ({\bf Quasilinear preference linear in parameter})  Consider the preference given by $u(t,q;\theta)=\theta q-t, \theta>0$. The indifference curves of this preference are linear and the slope of an indifference curve is given by $\frac{1}{\theta}$. Thus indifference curves of two distinct preferences can intersect at most once. Since $\theta\in ]0,\infty[$, $\{u(t,q;\theta)=\theta q-t\mid \theta>0\}$ is a rich single-crossing domain. We see that this preference is indeed linear, and thus depending in the context we refer to this preference to be linear.$\square$   
\label{ex:qlin}
\end{example}	

\begin{example}\rm({\bf A model with positive income effect}) Consider $\{u(t,q;\theta)=\theta q-t^{2}\mid \theta>0\}$.
	To see that this is a  single-crossing domain note that for any bundle with positive $t$ and positive $q$, the slope of an indifference curve is given by $\frac{2t}{\theta}$. Slopes are unequal for any two distinct values of $\theta$, therefore the single-crossing property is satisfied. The domain is rich because $\theta\in ]0,\infty[$.                
	This domain satisfies the notion of positive income effect introduced in \citep{Mishra1}.
	Let $\theta q^*-t_1^{'2}=\theta q^{**}-t_1^{''2}$, where $q^{*}<q^{**}, t_1'<t_1''$.
	Also let $\theta q^*-t_2^{'2}=\theta q^{**}-t_2^{''2}$, where $t_2'<t_2''$. 
	Also let $t_1'<t_2',t_1''<t_2''$.
	Then $\theta [q^{**}-q^{*}]=t_1^{''2}-t_1^{'2}=t_2^{''2}-t_2^{'2}$. 
	Thus $t_2''^2-t_1''^2=t_2'^2-t_1'^2$. Hence $[t_2''-t_1''][t_2''+t_1'']=[t_2'-t_1'][t_2'+t_1']$. 
	Now $[t_2''+t_1'']>[t_2'+t_1']$. Thus $t_2''-t_1''< t_2'-t_1'$. This entails
	$t_2''-t_2'<t_1''-t_1'$, which is required for positive income effect.$\square$ 	     
	\label{ex:index}
\end{example}

\noindent In Example \ref{ex:index} the domains is single-crossing, and also not quasilinear. The converse is also true. For example, $\{q^{\delta}-t\mid 0<\delta<1 \}$ is quasilinear and not single-crossing. At the bundle $(t=1,q=\frac{1}{8})$ the indifference curves of $q^{\frac{1}{3}}-t$ and $q^{\frac{2}{3}}-t$ are tangential. We note that
the indirect utility or the value function for these preferences are not maximum of affine functions in $\delta$. In Example \ref{example:delta} we construct a bounded CMS space consisting of indifference curves that are piece-wise differentiable to incorporate probability weighting functions $q^{\delta}$. In particular, there is a cutoff probability $q^*$ such that indifference curves are linear below the cutoff and strictly convex above the cutoff. 
It incorporates overweighting of large probabilities since $q^{\delta}$ lies above the $45^0$ ray from the origin for $q\in ]0,1[$.    
In Example \ref{example:Kahn} we consider the situation where the buyer underweight large probabilities. For certain parameter values the probability weighing function $\frac{\gamma q^\delta}{\gamma q^\delta+ (1-q)^{\delta}}$ in \citep{GonWu} can incorporate overweening large probabilities. In example \ref{example:delta} we consider a special case of probability weighing function suggested in \citep{GonWu} for probabilities close to $1$.

 \begin{example}\rm ({\bf Quasilinear preference non-linear in parameter}) 
	Let $\widehat{\theta}>0$ be a large valuation for the good and consider $U^{\delta}=\{\widehat{\theta} q^{\delta}-t\mid \delta \in [\frac{1}{4},\frac{1}{3}]\}$.   
	The slope of an indifference curve for a preference represented by $\widehat{\theta} q^{\delta}-t$
	is $\frac{1}{\widehat{\theta} \delta q^{\delta-1}}$.  Fix $q \in ]0,1[$ and Let $f(\delta)=\widehat{\theta}\delta q^{\delta-1}$. Then $f'(\delta)=\widehat{\theta} q^{\delta-1}[\delta\ln(q)+1]$. Let $q_{\frac{1}{3}}$ be such that 
	$f'(\frac{1}{3})=\widehat{\theta}(q_{\frac{1}{3}})^{\frac{1}{3}-1}[\frac{1}{3}\ln(q_{\frac{1}{3}})+1]=0$.  
	Since $(q_{\frac{1}{3}})^{\frac{1}{3}-1}>0$, the equality implies 
	$\frac{1}{3}\ln(q_{\frac{1}{3}})=-1$. Now $\frac{1}{3}\ln(q_{\frac{1}{3}})=-1\iff \ln(q_{\frac{1}{3}})=-3$. Thus $\ln(q_{ \frac {1}{4} })=-4$, and $q_{\frac{1}{4}}<q_{\frac{1}{3}}$ since $\ln$ is an increasing function. 
	Let $q^{*}$ be such that $q_{ \frac {1}{3}   }   <q^*<1$.
	Thus, for any given 
	$q\in [q^*,1]$, for all 
	$\delta\in [\frac{1}{4},\frac{1}{3}]$, $f'(\delta)>0$ since $\ln$ is an increasing function. 
	Then for $q\in [q^*,1]$, $\delta'' q^{\delta''-1}>\delta' q^{\delta'-1}$ if $\delta'<\delta''$.  
	Thus for any $q\in [q^*,1]$, 
	$\frac{1}{\widehat{\theta}\delta'' q^{\delta''-1}}<\frac{1}{\widehat{\theta}\delta' q^{\delta'-1}}$
	if $\frac{1}{4}\leq \delta'<\delta''\leq \frac{1}{3}$. That is in $q\in [q^*,1]$, 
	$U^{\delta}$ satisfies the single crossing property. Now let $g:[\frac{1}{4},\frac{1}{3}]\rightarrow [\frac{1}{8}, \frac{1}{2}]$ be a strictly increasing function which is a bijection. We use this bijection to create a single-crossing domain in the following way. 
		To begin with consider $\delta=\frac{1}{3}$. We have $g(\frac{1}{3})=\frac{1}{2}$. For every bundle $(t',q^*)$, where $q^*$ is defined in the preceding paragraph consider the indifference curves $\{(t,q)\mid \widehat{\theta}q^{ \frac{1}{3}}-t= \widehat{\theta}q^{* \frac{1}{3}}-t' \}$
	and $\{(t,q)\mid \frac{1}{2}q-t= \frac{1}{2}q^{*}-t' \}$. Define a preference
	$\overline{R}$ whose indifference curves for all $(t,q)$ with $q\geq q^*$ are given by 
	the ones from $\widehat{\theta} q^{\frac{1}{3}}-t$ and for $q\leq q^*$ are given by $\frac{1}{2}q-t$. 
	The preference $\overline{R}$ is monotone and continuous since both the upper contour sets and and lower contour sets are closed in $\mathbb{Z}$. Analogously append indifference curves of $\delta$ and $g(\delta)$. Since the slope of indifference curves of $\theta q-t$ decreases in $\theta$ we have
	obtain a compact single-crossing domain $[\underline{R},\overline{R}]$, where $\underline{R}$ is obtained for $\delta=\frac{1}{4}$ and $\theta=\frac{1}{8}$. We note that at $(t,q^*)$, 
	$\widehat{\theta}q^{*\delta}-t= g(\delta)q^*-t $ is not required to define the preferences.  
	Figure 0 depicts  
	the preferences $\overline{R}$, i.e., the preference corresponding to $\delta=\frac{1}{3}$, $\theta=\frac{1}{2}$ and $\underline{R}$, i.e., the preference corresponding to $\delta=\frac{1}{4}$, $\theta=\frac{1}{8}$.

	\begin{center}
		\includegraphics[height=6.5cm, width=10cm]{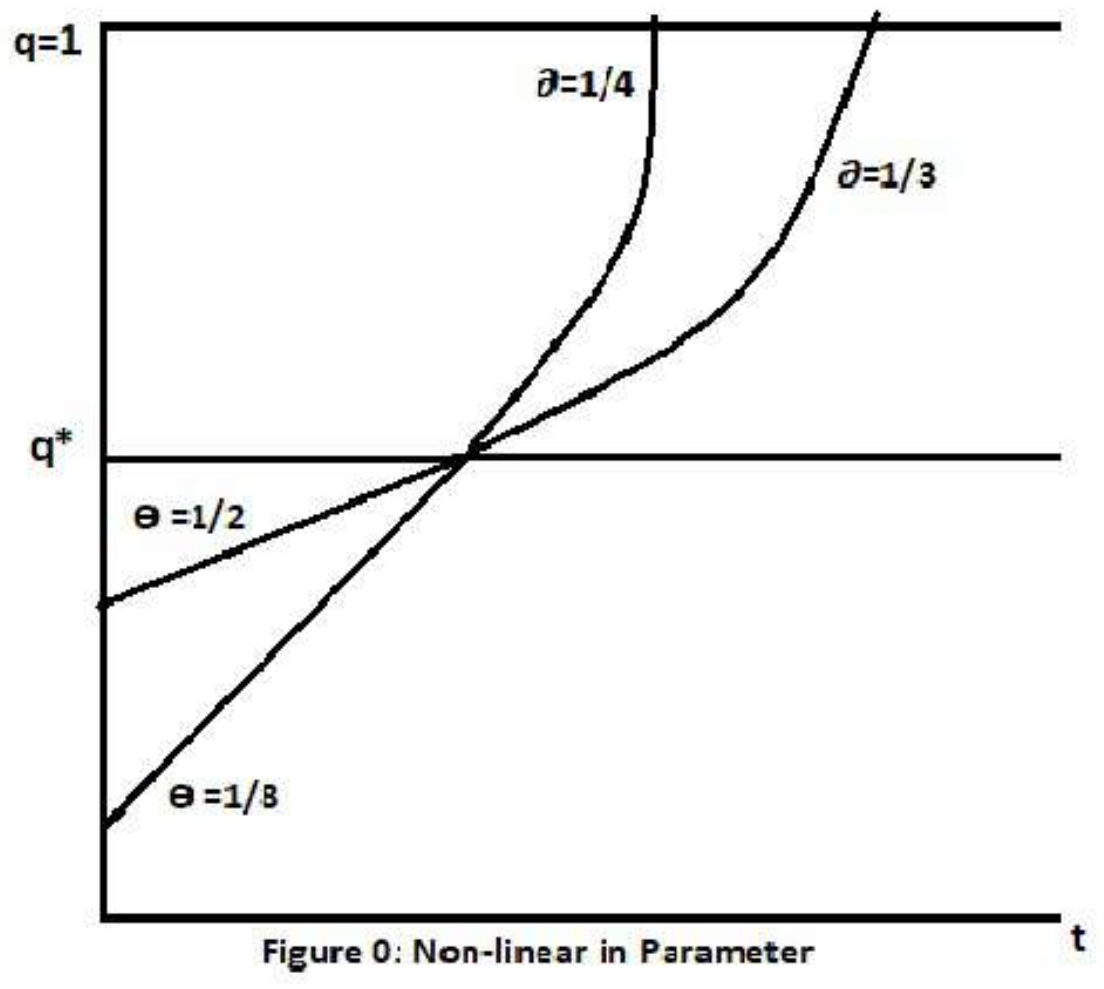}
	\end{center}  
	
	\label{example:delta}
\noindent In this example a high probability of win $q^*$ makes the buyer feel up beat about participating in the auction and in fact is willing to pay  $\widehat{\theta}$ which a high price. 	$\square$

\end{example}

\noindent The domain of preferences in Example \ref{example:delta} is  
parametrically two dimensional since we need $([\frac{1}{4},\frac{1}{3}],g)$ to specify a preference. To see that we need two parameters, note that it is not clear which preference alone the number $\frac{1}{2}$ refers to, $(\frac{1}{2}, g(\frac{1}{2}))$ or $(g^{-1}(\frac{1}{2}), \frac{1}{2})$.      
However these preferences are quasilinear. A domain of  multidimensional parametric, non quasilinear preferences is given in Example \ref{example:two-parameter}. In the next example we discuss underweighing of large probabilities.

\begin{example}\rm \citep{Kahn1} in their experiments found that agents underweight high probability and overweight low probability.
	\citep{Prelec} suggests probability weighing function $e^{- (-\ln q)^{\delta}  }, 0<\delta<1$ and \citep{GonWu} suggest $\frac{\gamma q^\delta}{\gamma q^\delta+ (1-q)^{\delta}}$ that incorporate this idea in \citep{Kahn1}. 
	Whether buyers  who underweight high probability will also like to pay less compared with the ones who overweight high probability maybe an important experiential question. From a technical perspective it is possible to incorporate overweighting of large probability. Consider $[q^*, 1]$, $q^*\neq 1$ close to $1$ and   
	$e^{- (-\ln q)^{\delta}  }, 0<\delta<1$. We can choose $q^*$
	such that for any $q\in [q^*, 1]$, derivative $\frac{\partial e^{- (-\ln q)^{\delta}  }}{\partial q}$ is decreasing in $\delta$. To see this note, 
	$\frac{\partial^{2}}{\partial q \partial \delta}\left(e^{- \left(- \ln{\left(q \right)}\right)^{\delta}}\right)$
	
	$=\frac{\delta \left(- \left(- \ln{\left(q \right)}\right)^{\delta} \left(- \ln{\left(q \right)}\right)^{\delta - 1} e^{- \left(- \ln{\left(q \right)}\right)^{\delta}} \ln{\left(- \ln{\left(q \right)} \right)} + \left(- \ln{\left(q \right)}\right)^{\delta - 1} e^{- \left(- \ln{\left(q \right)}\right)^{\delta}} \ln{\left(- \ln{\left(q \right)} \right)}\right) + \left(- \ln{\left(q \right)}\right)^{\delta - 1} e^{- \left(- \ln{\left(q \right)}\right)^{\delta}}}{q}$

	\noindent Now consider the numerator 
	
	$ e^{- \left(- \ln{\left(q \right)}\right)^{\delta}} \left(- \ln{\left(q \right)}\right)^{\delta - 1}\Big(- \delta \left(- \ln{\left(q \right)}\right)^{\delta} \ln{\left(- \ln{\left(q \right)} \right)} + \delta\ln{\left(- \ln{\left(q \right)} \right)} +1\Big)$.
	
	\noindent The term outside the parenthesis is positive, so we consider the term inside the parenthesis which is,        
	
	$=\delta\ln{\left(- \ln{\left(q \right)} \right)} \Big (- \left(- \ln{\left(q \right)}\right)^{\delta}+1\Big) +1$. For $q$ close to $1$, $\delta\ln{\left(- \ln{\left(q \right)} \right)}$ is a large negative number. Further, for $q$ cose to $1$, 
	$\Big (- \left(- \ln{\left(q \right)}\right)^{\delta}+1\Big)$ is positive. Then we can extend the preference to all bundles in $\mathbb{Z}$ as done in Example \ref{example:delta}.$\square$

	\label{example:Kahn}  	
\end{example}

\noindent Analogous to Example \ref{example:delta}, Example  \ref{example:two-parameter} demonstrates that single-crossing domains can incorporate multidimensional parametric classes of preferences. 
However, behaviorally these two examples are different. Example  \ref{example:two-parameter} is more about how buyers with higher valuations value loss of utility from paying for the good relative to the buyers with lower valuations.    
Example \ref{example:two-parameter} also incorporates positive income effects.

\begin{example}\rm ({\bf A non quasilinear model of multidimensional types}) Consider the following sets: $U=\{u(t,q)=\theta\sqrt{q}-t^{2}|\theta\in ]0,2]\}$, and $V=\{v(t,q)=2\sqrt{q}-\alpha t^{2}|\alpha\in ]0,1]\}$. We show that $U\cup V$ is single-crossing and rich. 
	We show that $U\cup V$ is a single crossing domain. In the interior of $\mathbb{Z}$ slopes of indifference curves of preferences  from $U$ are given by $\frac{4t\sqrt{q}}{\theta}$. 
	If $\frac{4t\sqrt{q}}{\theta'}=\frac{4t\sqrt{q}}{\theta''}$, then $\theta'=\theta''$. Slopes of indifference curves of  preferences from $V$ is given by $2\alpha t\sqrt{q}$. Analogously, if $2\alpha^{'}t\sqrt{q}=2\alpha^{''}t\sqrt{q}$, then $\alpha^{'}=\alpha^{''}$. Hence, within $U$ and $V$ preferences satisfy the single-crossing property. We need to show that the single-crossing property holds across $U$ and $V$. If single-crossing did not hold across $U$ and $V$, then for some $(t,q)$ in the interior of $\mathbb{Z}$ and for some  for some $\theta,\alpha$ we have $\frac{4t\sqrt{q}}{\theta}=2\alpha t\sqrt{q}$. 
	This implies $\alpha=\frac{2}{\theta}$. Since $\theta\in ]0,2]$, we must have $\alpha\geq 1$. Since   $\alpha\in ]0,1]$, $\alpha=\frac{2}{\theta}$ holds for $\theta=2$ and $\alpha=1$. However, for these values of $\theta=2$ and $\alpha=1$ we obtain a only one preference which is $2\sqrt{q}-t^{2}$, i.e.,  $U \cap V=\{2\sqrt{q}-t^{2}\}$.  
	Thus, $U \cup V$ satisfies the single crossing property.

	We show that $U\cup V$ is rich.  
	To see this consider expanding $U$ and $V$, i.e., $U^{*}=\{u(t,q)\mid \theta\sqrt{q}-t^{2}|\theta\in ]0,\infty[\}$ and $V^{*}=\{v(t,q)\mid 2\sqrt{q}-\alpha t^{2}|\alpha\in ]0,\infty[\}$.
	Consider $z=(t,q),z'=(t',q')$ diagonal such that $q<q^{'}$, and $t<t'$. Set $\theta=\frac{t^{'2}-t^{2}}{\sqrt{q}'-\sqrt{q}}$, and $\alpha=\frac{2\sqrt{q}'-2\sqrt{q}}{t^{'2}-t^{2}}$, thus $U^{*}$ and $V^{*}$ are rich.  To show that $U\cup V$ is rich we need to show that $z,z'$ lie on either an indifference curve from $U$ or from $V$. Since $U\cup V$ is single-crossing $z,z'$ can lie on an indifference curve of a preference from either  $U$ or $V$ but not both other than the preference in the intersection. To complete the proof that $U\cup V$ is rich, we show that if $z,z'$ do not lie on an indifference curve from a preference from $U$, then  
	$z, z^{'}$ do not lie on an indifference curve of a preference from $2\sqrt{q}-\alpha t^{2}, 1<\alpha$.
	Let by way of contradiction, $1<\alpha=\frac{2\sqrt{q}'-2\sqrt{q}}{t^{'2}-t^{2}}$. Then, $\frac{t^{'2}-t^{2}}{ 2\sqrt{q}'-2\sqrt{q} }<1$, i.e., $\frac{t^{2}-t^{2}}{ \sqrt{q}'-\sqrt{q}  }<2$. This means there exists $\theta<2$ such that $z,z'$ are in the same indifference curve of a preference from $U$, this is a contradiction. Therefore, our claim is established, and given that $V^{*}$ is rich $z,z'$ lie on an indifference curve from $V$.    
	Hence, $U\cup V$ is rich. For high enough benefits, i.e., when  $\theta=2$, the utility functions with $\alpha<1$ describes that disutility from 
	per unit cost is less compared to the situation when the benefit is not high enough i.e., $\theta<2$.

	We do not have one dimensional pay-off relevant parametric representation of $U\cup V$.  Let us ask whether $2$ represents  $2\sqrt{q}-t^{2}$  or $2\sqrt{q}-\frac{1}{2}t^{2}$. Suppose we break the tie by  letting $2$ represent $2\sqrt{q}-t^{2}$ and $\frac{1}{2}$ represents $2\sqrt{q}-\frac{1}{2}t^{2}$. Then the question is how do we represent $\frac{1}{2}\sqrt{q}-t^{2}$. The only option available is $1$ since $1$ multiplies $t^2$. Then $1$ cannot represent $\sqrt{q}-t^{2}$. Thus we need two parameters to represent $U\cup V$. It is not critical that two parameters appear in the preferences in $U\cup V$, rather the critical point is that none of the parameters are redundant.
	If valuations are required to be pay-off relevant parameters that appear in the utility functions, then the real number that is associated with a preference in $\mathcal{R}^{cms}$ cannot be called a valuation.
	First, as we have seen, a rich single-crossing domain being one dimensional topological manifold does not 
	mean that it is also parametrically one dimensional. Also, since 
	any open interval is homeomorphic to the real line there is nothing special about the homeomorphism defined in Theorem \ref{thm:lin_cont}. For example $\mathcal{R}^{cms}$ is homeomorphic to both $]0,1[$ and $]7, 8[$, and thus it is not clear which interval should represent valuation. 
	We note that $U\cup V$ is not convex.{\footnote{A set of real valued functions is convex if for any two functions $u',u''$ from the set, the function $\lambda u'+(1-\lambda )u'', 0\leq \lambda \leq 1$ is in the set. The addition $\lambda u'+(1-\lambda )u''$ is defined pointwise.}}$\square$

	\label{example:two-parameter}
\end{example}

\noindent We discuss strategy-proof mechanisms in the next section.

\section{\bf Strategy-proof Mechanisms in $\mathcal{R}^{cms}$}
\label{sec:sp}

Consider a domain $\mathcal{R}^{cms}$. We study mechanisms or social choice function that maps reported 
preferences from $\mathcal{R}^{cms}$ to a bundle. We define strategy-proofness next.

\begin{defn} [Strategy-proof Mechanisms]\rm A function $F:\mathcal{R}^{cms}\rightarrow \mathbb{Z}$ is called a {\bf mechanism}. A mechanism is {\bf strategy-proof} if for every $R', R''\in \mathcal{R}^{cms}$,
	$F(R')R'F(R'')$.   
	\label{defn:sp}
\end{defn}

\noindent If $R'$ is the buyer's actual preference, and she reports $R''$, then $F(R')R'F(R'')$ means that the buyer has no incentive to misreport her preference since the bundle that she is allocated by $F$ at $R''$, i.e., $F(R'')$, does not make her better off relative to the bundle $F(R')$ which she is allocated at $R'$.  
We also consider a local version of the notion of strategy-proofness. 
The set $Rn(F)=\{z\in \mathbb{Z}\mid F(R)=z,~\text{for some}~R\in \mathcal{R}^{cms}\}$ denotes the range of $F$. 

\begin{defn}\rm For a mechanism $F$ we call $Rn(F)$ to be {\bf ordered} if and only if for all $x',x''$ in $Rn(F)$ with  $x'\neq x''$ implies  $x'$ and  $x''$ are diagonal, i.e., either $x'<x''$ or $x''<x'$.     
\end{defn} 

\noindent We define the notion of local strategy-proofness next.

\begin{defn}\rm For the mechanism $F$ let $Rn(F)$ be ordered. Let $F(R')=x',F(R'')=x'' $ and $x'<x''$. Further, let for all $z\in Rn(F)$ with $z\neq x',z\neq x''$ either $z<x'$ or $x''<z$. 
	In this case we call $x', x''$ to be {\bf adjacent}. We say $F$ is {\bf locally strategy-proof in range}
	if and only if for all adjacent bundles $F(R'),F(R'')$, $F$ satisfies $(i)$ $F(R')R'F(R'')$ and $(ii)$ $F(R'')R''F(R')$.    	 
	\label{defn:sp_local}
\end{defn}

\noindent Local strategy-proofness in range requires the buyer to be unable to obtain a better adjacent bundle by misreporting her preference.
We show that if $F$ is strategy-proof, then $Rn(F)$ is an ordered set. Therefore,  if $F$ is strategy-proof, then $F$ is locally strategy-proof in range.  
A different notion of local strategy-proofness is often used in the studies of the implications of local strategy-proofness for strategy-proofness. 
\citep{Carroll} is the first paper to define this notion by means of adjacent preferences. Since our paper is not about characterizations of local strategy-proof mechanisms, we do not discuss it further.
Further, we do not assume that $F$ is local strategy-proof in range, rather it is a consequence of monotonicity of $F$ and continuity of the corresponding indirect preference correspondence defined below.        
The two component functions of $F$ are also denoted by $t,q$, i.e.,  
$F(R)=(t(R),q(R))$. The following is the standard equivalent formation of the notion of an ordered range in terms of monotonicity of $F$.

\begin{defn}\rm 
	A SCF $F:\mathcal{R}^{cms}\rightarrow \mathbb{Z}$ is \textbf{monotone with respect to the order relation} $\prec$ on $\mathcal{R}^{cms}$ or simply \textbf{monotone} if for every $R', R''\in \mathcal{R}^{cms}$,
	[$R'\prec R''\iff F(R') \leq F(R'')$].
	\label{defn:mon}
\end{defn}

\noindent The following lemma shows that if $F$ is strategy-proof, then $F$ is monotone, i.e.,
$Rn(F)$ is an ordered set for strategy-proof mechanisms.       

\begin{lemma}\rm 
	Let  $F:\mathcal{R}^{cms}\rightarrow \mathbb{Z}$ be a strategy proof SCF. Then $F$ is monotone.
	\label{lemma:mon}
\end{lemma}

\begin{proof} See Appendix $2$. 
	
\end{proof}	 

\noindent A strategy-proof mechanism induces an indirect utility function if an utility representation of a preference is given. In this paper we study strategy-proof mechanisms by analyzing  ordinal properties of preferences. Thus, instead of an indirect utility function we consider an indirect preference correspondence. This correspondence is defined next.

\begin{defn} [Indirect Preference Correspondence] \rm 
	Consider a mechanism $F$. We call the set valued mapping, $V^{F}:\mathcal{R}^{cms}\rightrightarrows \mathbb{Z}$ defined as $
	V^{F}(R)=IC(R, F(R))=\{z\in \mathbb{Z}\mid z I F(R)\}$ {\bf indirect preference correspondence}. 
\end{defn}

\noindent For each preference $R$, $V^{F}$ picks the set of bundles $z$ that are indifferent to $F(R)$ under $R$. For $R',R''$ by $R'\precsim R''$ we mean either $R'=R''$ or $R'\prec R''$. 
Since $\mathcal{R}^{cms}$ has the least upper bound property, any bounded monotone, i.e., either increasing or decreasing, sequence $\{R^{n}\}_{n=1}^{\infty}$ converges to its supremum if the sequence is increasing and converges to its infimum if the sequence is decreasing. The following notion of continuity of an indirect preference correspondence is crucial for the analysis of strategy-proof mechanisms.        

\begin{defn}[Continuous Indirect preference correspondence]\rm We call
	$V^{F}:\mathcal{R}^{cms}\rightrightarrows \mathbb{Z}$ {\bf continuous}, if for any $R$ and any monotone sequence $\{R^{n}\}_{n=1}^{\infty}$ converging to $R$: (a) the sequence $\{F(R^n)\}_{n=1}^{\infty}$ 
	converges and, (b) $z=\lim_{n\rightarrow \infty}F(R^n)I F(R)$.
	\label{defn:cont}
\end{defn}

\noindent Suppose $U$ is a utility representation of $R\in \mathcal{R}^{cms}$. Further suppose $V^{F}(U)=U(F(U))$. Then the notion of continuity of $V^{F}$ becomes the standard notion of continuity of a function.
That is, given that the space of utility functions that represents the preferences is a metric space,
$V^{F}$ is continuous if $U^{n}\rightarrow U$, then $V^{F}(U^{n})\rightarrow V^{F}(U)$. 
Consider Example \ref{ex:qlin}. Let $F:]0,\infty[\rightarrow \mathbb{Z}$ be a mechanism. 
Then $V^{F}(\theta)=\theta q(\theta)-t(\theta)$ is a function, and thus instead of a continuous
indirect preference correspondence we have a continuous indirect utility function. We do not assume $V^{F}(\theta'')=V^{F}(\theta')+\int_{\theta'}^{\theta''}q(r)dr$ for $\theta'<\theta''$ as a sufficient condition in our characterization of strategy-proof mechanisms even for CM rich single-crossing domains where such an equation is well defined. 
We only require continuity of $V^{F}$ to study preferences in Example \ref{ex:qlin}. 
Next we show that if $F$ is a strategy-proof mechanisms, then  the correspondence $V^{F}$ is continuous.
The following intermediary result is important.

\begin{lemma}\rm Let $z',z''\in \mathbb{Z}$ and $z'<z''$. Let $R\in \mathcal{R}^{cms}$. 
	Then the following hold: $(i)$ if $z'I z''$, then $z'P^{*}z''$ for all $R^{*}\prec R$; and $z''P^{*}z'$ for all $R\prec R^{*}$,  
	$(ii)$ if $z'Pz''$, then $z'P^{*}z''$ for all $R^{*}\prec R$; and   
	if $z''Pz'$, then $z''P^{*}z'$ for all $R\prec R^{*}$, where $R^{*}\in \mathcal{R}^{cms}$.   
	\label{lemma:preference_preserve}
\end{lemma}	

\begin{proof} See Appendix $2$.
	
\end{proof}         

\noindent Lemma \ref{lemma:preference_preserve}, the second part of the lemma in particular, implies that the order on the set of preferences is preserved over diagonal bundles. The notions of single-crossing property defined in \citep{Saporiti} and \citep{Baisa2} are analogous to the second part of Lemma \ref{lemma:preference_preserve}.

\begin{lemma}\rm 
	Let $F:\mathcal{R}^{cms}\rightarrow \mathbb{Z}$ be strategy-proof. Then $V^F:\mathcal{R}^{cms}\rightrightarrows \mathbb{Z}$ is continuous.
	\label{lemma:cont_correspondence}
\end{lemma}

\begin{proof} See Appendix $2$. 
	
\end{proof}

\noindent We can summarize the observations from Lemma \ref{lemma:cont_correspondence} and Lemma \ref{lemma:mon} in the following theorem.

\begin{theorem}\rm Let $F:\mathcal{R}^{cms}\rightarrow \mathbb{Z}$ be strategy-proof. Then $F$ is monotone, and $V^{F}$ is continuous. Further, $F$ is locally strategy-proof in range.    
	\label{thm:sp_implies}	     
\end{theorem}

\noindent To prove the converse of Theorem \ref{thm:sp_implies} we require to put some restrictions on the range of $F$.  
To prove the converse of Theorem \ref{thm:sp_implies} we assume that $Rn(F)$ is finite.
Next we give  two examples that demonstrate that our two axioms are independent.       
In the next example we construct mechanism where $V^F$ is continuous and $F$ is not monotone.  

\begin{example}\rm Consider the quasilinear preferences $\{u(t,q;\theta)=\theta q-t\mid \theta >0\}$. 
	Let $\underline{\theta}-t^1=0$, and $\overline{\theta}-t^2=0$. Let $\underline{\theta}<\overline{\theta}$. Thus, $t^1<t^2$. Let $F(\theta)=(t^2,1)$ if $\overline{\theta}\leq \theta$, $F(\theta)=(0,0)$, if $\theta\in ]\underline{\theta},\overline{\theta}[$, $F(\theta)=(t^1,1)$ if $\theta<\underline{\theta}$.
	Then $V^{F}(\theta)=\theta-t^2$ if $\overline{\theta}\leq \theta$, $V^{F}(\theta)=0$ if $\theta\in ]\underline{\theta},\overline{\theta}[$ and  $V^{F}(\theta)=\theta-t^1$ if $\theta\leq \underline{\theta}$. Then $V^{F}$ is continuous, and $F$ is not monotone.$\square$
	\label{ex:non_mon} 
\end{example}

\noindent In the next example we construct mechanism where $V^F$ is not continuous and $F$ is monotone.  

\begin{example}\rm Consider the quasilinear preferences $\{u(t,q;\theta)=\theta q-t\mid \theta >0\}$. 
	Let $(t^1,q^1)<(t^2,q^2)$. Let $\theta q^1-t^1=\theta q^2-t^2$. 
	Let $\theta^*<\theta$. 
	Let $F(\theta)=(t^1,q^1)$ if $\theta<\theta^*$ and $F(\theta)=(t^2,q^2)$ if $\theta^*\leq \theta$. 
	Thus, $F$ is monotone. To see that $V^F$ is not continuous consider a sequence $\theta^n\uparrow \theta^*$. $\lim_{n\rightarrow \infty} F(\theta^n)=(t^1,q^1)$. Thus, $\lim_{n\rightarrow \infty}\theta^n q^1-t^1=\theta^* q^1-t^1> \theta^* q^2-t^2$. 
	Hence $V^F$ is not continuous.$\square$   	
	
	\label{ex:non_cont}	
\end{example}

\subsection{\bf Analysis of Strategy-proofness when $Rn(F)$ is Finite}              

\begin{theorem}\rm Let $F: \mathcal{R}^{cms}\rightarrow \mathbb{Z}$ be monotone, and $V^F$ be continuous. Further, let $Rn(F)$ be finite. Then $F$ is strategy-proof.
	\label{thm:implies_strtagey_proof_finite_range}	
\end{theorem}	

\begin{proof} First we show that  monotonicity of $F$ and continuity of $V^{F}$ imply that $F$ is locally strategy-proof in range. Then using the single-crossing property strategy-proofness is extended to $Rn(F)$.   See See Appendix $2$ for details.  
	
\end{proof}

\noindent The proof of Theorem \ref{thm:implies_strtagey_proof_finite_range} that we present is constructive, i.e., the proof reveals how any strategy-proof mechanism looks like. We consider an example of such a rule to illustrate the geometry. Suppose $\#Rn(F)=3$, $\#$ denotes the number of elements in a set. Let $Rn(F)=\{a,b,c\}$, and let $a<b<c$. The proof of Theorem \ref{thm:implies_strtagey_proof_finite_range} shows that if a $F$ is monotone and $V^{F}$ is continuous, then $F$ must be defined as follows.

$$F(R) = \begin{cases}
a, & \text{if $R \prec R_{1}$;}\\
\text{either}~ a~\text{or}~ b, & \text{if $R=R_1$}\\
b & \text{if $R_1\prec R \prec R_2$}\\
\text{either}~b~\text{or}~c & \text{if $R=R_2$,}\\
c & \text{if $R_2\prec R$}
\end{cases}$$

\noindent where $aI_1b$ and $bI_2c$. To see that $F$ is strategy-proof, without loss of generality consider $R$ such $F(R)=b$. By Lemma  \ref{lemma:preference_preserve}, $bPa$ and  
$bPc$. In Figure $1$ we depict this mechanism geometrically. The indifference curves that are drawn in Figure $1$ are convex for the sake of simplicity. We do not require convexity of preferences in our proofs.  

\begin{center}
	\includegraphics[height=5.5cm, width=12cm]{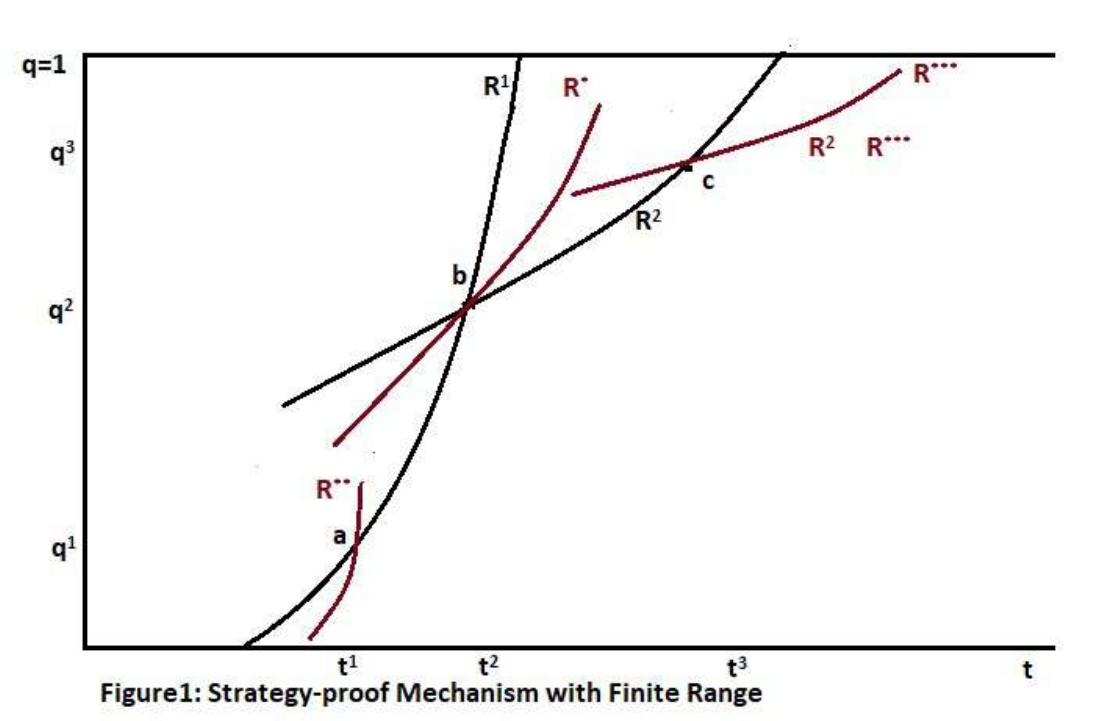}
\end{center}

\noindent In Figure $1$, $a=(t^1,q^1), b=(t^2,q^2),c=(t^3,q^3)$, $R^{**}\prec R^{1}, R^1\prec R^{*}\prec R^{2}$ and $R^{2}\prec R^{***}$. Another possible representation of $F$ with three elements is when all three bundles fall on a single indifference curve, 
i.e., $aIbIc$ for some $R$. By the single-crossing property $F(R)=b$.  
The crucial aspect of our constructive proof is that given a set of the bundles, i.e.,
given their ``coordinates'', monotonicity of $F$ and continuity of $V^{F}$ fix the preferences under which at least two, and maximum three of these, bundles are indifferent. By the single-crossing property these preferences are unique. 
Once these special preferences are obtained the only task that remains is to fix the allocations for preferences in between any two such special preferences. The special preferences in Figure $1$ are $R^1$ and $R^2$. Further, we can easily read off the formula to compute the expected revenue from this geometry. Let $\mu$ be a probability measure that is defined on the intervals from $\mathcal{R}^{cms}$. 
The formal definition of this measure is provided in Section \ref{sec:opt}. 
For example, the expected revenue from the mechanism in Figure $1$ is given by
$t^1\mu(]-\infty,R^{1}])+t^2\mu(]R^{1},R^{2}])+t^{3}\mu([R^{2},\infty[)$, where $\mu(\{R\})=0$ for all $R\in \mathcal{R}^{cms}$.

\begin{remark}\rm By the ``geometry of strategy-proofness'', we mean the geometry displayed in Figure $1$.      
\end{remark}

\noindent Next we make a remark about ``monotonicity'' condition that is used characterization of strategy-proof mechanisms. 

\begin{remark}\rm \citep{Mishra1} study dominant strategy incentive compatible (DSIC) mechanisms and uses  a generalized version of the weak monotonicity condition studied in \citep{Bik}. The latter considers quasilinear preferences. We try to see how the monotonicity condition in \citep{Mishra1} looks like in the context of our model.
	Let $F(\widehat{R})=(t',q')$, and $\widehat{R}\prec R$. Let $F(R)=(t'',q'')$. Let $q''> q'$. Then,  $t''>t'$ and $(t'',q'')I(t',q')$. According to the the notations in \citep{Mishra1} 
	$z=(t(\widehat{R}), q(\widehat{R}))$, $a=q''$, $V^{R}(q'',z)=t''$. Thus if $(t',q'), (t'',q'')$ are adjacent bundles $R$ is one of the special preference that we identify in our characterization. 
	\citep{Mishra1} provide an extensive literature review on the ``monotonicity'' condition used in the characterization of DSIC mechanisms for quasilinear preferences.        	     
	
	\label{remark:mon}	
\end{remark}

\noindent We end this section with the definition of an individually rational mechanism.

\begin{defn}\rm A mechanism $F:\mathcal{R}^{cms}\rightarrow \mathbb{Z}$ is {\bf individually rational}
	if and only if for all $R$, and for all $t$, $F(R)R(t,0)$.     
\end{defn}	   

\noindent The bundle allocated by the mechanism $F$ does not make the buyer worse off relative to any bundle
that has $q=0$. By monotonicity of $R$, $F(R)R(0,0)\implies F(R)R(t,0)$. Thus without loss of generality we shall consider a mechanism to be individually rational if $F(R)R(0,0)$.  We end this section with an example which demonstrates that if a mechanism has infinite range, then our axioms are not enough to guarantee strategy-proofness.

\begin{example}\rm ({\bf A mechanism with infinite range which is not strategy-proof}) 
	We consider 
	a set of preferences given by $\{u(t,q;\theta)=\theta q-t, \theta \in [1,2]\}$. Consider a line segment that connects
	$(0, 0)$ and $(\frac{1}{3},1)$. The equation of this line is $q = 3t$.
	Let $F:[1,2]\rightarrow \mathbb{Z}$ be defined as $F(\theta)=(\frac{1}{3}\theta-\frac{1}{3}, \theta-1)$. 
	Then $Rn(F)=\{(t,q)\mid q=3t, t\in [0,\frac{1}{3}]\}$ is a continuum. The mechanism $F$ is monotone and continuous. Further, $F$ is individually rational and not strategy-proof. 
	To see that $F$ is individually rational, note that $u(F(\theta);\theta)=\theta(\theta-1)-(\frac{1}{3}\theta-\frac{1}{3})=\theta(\theta-1)-\frac{1}{3}(\theta-1)=(\theta-1)(\theta-\frac{1}{3})\geq 0$ since $\theta\geq 1$.
	To see that $F$ is not strategy-proof note that the slope of indifference curves of any
	preference in the domain of $F$ lies in $[ \frac{1}{2},1]$. This slope is smaller than the slope of the line segment $q=3t$. Thus the maximizations of utility for all preferences in the domain occur at $(t,q)=(\frac{1}{3},1)$. Now $V^{F}(\theta)=u(F(\theta);\theta)=(\theta-1)(\theta-\frac{1}{3})$ is continuous. Further, $\frac{dV(\theta)}{d\theta}=(\theta -\frac{1}{3})+(\theta- 1)\neq (\theta-1)=q(\theta)$. Also $\frac{d^{2}V(\theta)}{d\theta^{2}}=2>0$. 
	Thus the indirect utility function is convex. 
	The crucial point here is that the derivative of the indirect utility function
	is not equal to the allocation probability. Therefore this example does not satisfy the integral condition discussed in \citep{myer}. Our analysis shows that situations in which the range of a mechanism is finite we do not need to assume the integral condition to characterize strategy-proof mechanisms.$\square$         
	\label{ex:not_sp}
\end{example}

\begin{remark}\rm Example \ref{ex:not_sp} is not in conflict with the standard mechanism design approach where given a monotone allocation rule, i.e., monotone $q$, a payment rule $t$ is constructed so that $F=(t,q)$ is a strategy-proof mechanism. Example \ref{ex:not_sp} is saying that an arbitrary monotone mechanism is not strategy-proof, there should be some relationship between the allocation and payment rule for the mechanism to be strategy-proof. One such well known relationship is given by the revenue equivalence equation for the preferences in Example \ref{ex:qlin}.  
\end{remark}	

\subsection{\bf Optimal Mechanisms when $Rn(F)$ is Finite}
\label{sec:opt}

We study optimal mechanisms for intervals $[\underline{R},\overline{R}]$ as a subspace of $\mathcal{R}^{cms}$ in the order topology. In the proof of Theorem \ref{thm:implies_strtagey_proof_finite_range}, Lemma \ref{lemma:mon} and Lemma \ref{lemma:cont_correspondence} we have made observations that these results hold for 
closed intervals as well.   
Let the Borel sigma-algebra on $[\underline{R},\overline{R}]$ generated by the subspace order topology    
be denoted by ${\cal B}$. Let $\mu$ denote a probability measure on ${\cal B}$. Thus, $([\underline{R},\overline{R}],{\cal B}, \mu)$ is a probability space. We assume $\mu $ to be $(i)$ ({\bf non-trivial}) for any non-empty and non-singleton interval $A\subseteq [\underline{R},\overline{R}]$ that is  $\mu(A)>0$ ; and $(ii)$ ({\bf continuous}) for any $A\in {\cal B}$ with $\mu(A)>0$ and any $c\in \Re$ with $0<c<\mu(A)$, there exists $B\in {\cal B}$ with $B\subseteq A$ and $\mu(B)=c$. Continuity of $\mu$ implies $\mu$ is non-atomic, in particular $\mu(\{R\})=0$ for all $R\in [\underline{R},\overline{R}]$. We assume the real line $\Re$ to be endowed with the standard Borel
sigma algebra obtained from the Euclidean metric space. We denote this sigma-algebra by $B(\Re)$.
We show that if $F$ is strategy-proof, then the component functions $t, q $ are ${\cal B}/B(\Re)$ measurable.

\begin{lemma}\rm Let $F:[\underline{R},\overline{R}]\rightarrow Rn(F)$ be monotone and 
	$Rn(F)$ be finite. Then the component functions of $F$, $t,q$ are $\mathcal{B}/B(\Re)$measurable.      
	\label{lemma:measurable}
\end{lemma}	

\begin{proof}  See Appendix $2$.          	
\end{proof}

\noindent Now we can define expected revenue of the seller.

\begin{defn}\rm Let $[\underline{R},\overline{R}]$ and $\mu$ be given. Let $Rn(F)$ be finite. The expected revenue form $F$ is $\int_{[\underline{R},\overline{R}]}t(R)d\mu(R)$.
	We denote the expected revenue from $F$ by $E[F]$. 
	Since we fix a probability measure we do not mention $\mu$.       	      
\end{defn}

\noindent The important aspect of the optimal mechanism is the constraint set. Thus without loss of generality we shall assume that revenue of the seller is given by $t$.   
We define optimal mechanism next. 

\begin{defn}\rm Let $[\underline{R},\overline{R}]$ and $\mu$ be given. Let 
	\[\mathcal{F}=\{F: [\underline{R},\overline{R}]\rightarrow \mathbb{Z}\mid F~\text {is strategy-proof, indvidually rational, has finite}~Rn(F)\}.\]
	Then $F^{*}\in \mathcal{F}$ is {\bf optimal} if and only if for all $G\in \mathcal{F}$, $E[F^{*}]\geq E[G]$. 
	
	\label{defn:optimal} 
\end{defn}

\noindent Throughout our discussion of optimal mechanisms we assume $\mu$ to be continuous and non-trivial. 
The set $\mathcal{F}$ is a subspace of $\mathbb{Z}^{[\underline{R},\overline{R}]}$. Instead of looking for an optimal mechanism in $\mathbb{Z}^{[\underline{R},\overline{R}]}$ 
we look for an optimal mechanism in $\mathcal{F}_{l}=\{F_{l}\mid F_{l}\in \mathcal{F}, \#Rn(F)\leq l\}$. By strategy-proofness the bundles in the range of any $F$ are ordered, and since $\#Rn(F)\leq l$  we can embed each $F_{l}$ in $\mathbb{Z}^{l}$.
The the problem of finding an optimal mechanism in $\mathcal{F}_{l}$ is computationally easier compared with the problem of finding one in $\mathbb{Z}^{[\underline{R},\overline{R}]}$. We show that the optimal mechanism in each $\mathcal{F}_{l}$ exists. To show the existence, we consider an optimization program that naturally follows from the geometry of strategy-proofness due to the single-crossing property. The constraint set of the optimization program contains elements from the topological space $[\underline{R}, \overline{R}]$. However, we believe that this may not be a major problem given that computations on abstract topological spaces is now a very active area of research and the optimization program should be solvable. After obtaining an optimal solution, analytical or approximate, we can compare the revenues from the optimal mechanism $F_{l}^{*}$ in $\mathcal{F}_{l}$ for various $l$s. If for large enough sample of values of $l$s the revenue does not change much, i.e., lies in acceptable bound then we may have an approximately optimal mechanism in $\mathcal{F}$. As an example we solve 
the optimal mechanism for quasilinear preferences by assuming monotone hazard rate. It turns out that in this problem 
for every $l$, $F_{l}^*$ can have at most two distinct bundles in the range, i.e., the optimal mechanism 
for every $l$ is deterministic. That is, in this particular case the sequence of optimal mechanism $\{F_{l}^{*}\}_{l=2}^{\infty}$ is a constant and thus converges immediately. In general, what happens in a specific problem  depends on the specificities of the  problem. For some problems we may not even obtain an analytical solution. Our approach provides an optimization program which is not specific to any particular problem. For example in \citep{myer} the payment equation obtained due to revenue equivalence, latter is due to strategy-proofness for an one agent problem, is replaced in the expected revenue formula; a solution to the reduced problem entails an optimal mechanism and thus shows the existence as well.            
Substituting the payment function obtained from the 
revenue equivalence equation in the expected revenue function may not  make the problem computationally easy. 
For example if there is a discrepancy between the way the buyer evaluates payment and the way seller does computational difficulties may arise.  
Suppose the buyer's utility function is $u(t,q;\theta)=\theta q-t^2$, where $q$ denotes probability of win, but the seller's revenue is $t$.           
\citep{Tian} make a note about verifying the regularity conditions required for a general envelope theorem.    
We show the existence of an optimal mechanism without using revenue equivalence. 
Given the problem of finding the optimal mechanism with a fixed maximum number of elements in the range we may ask three questions: (1) what optimization problem to solve, (2) how to solve the optimization problem and (3) what are the solutions to the optimization problem. This paper addresses the first question. How to solve depends on the optimization techniques and also depend on whether the preferences under consideration have a representation in terms of a parametric class or not. 
What exactly the solutions are certainly depend on the specificity of the problem at hand.  
Now we proceed to study the optimal mechanism in $\mathcal{F}_{l}$.

\begin{theorem}\rm Let $\mu$ be non-trivial and continuous. For every non negative integer $l$, the optimal mechanism $F_{l}^*$ exists where  $F_{l}^*\in \mathcal{F}_{l}=\{F_l\mid F_l\in \mathcal{F}, \#Rn(F)\leq l\}$ and $\mathcal{F}$ is as defined in Definition \ref{defn:optimal}.
	\label{thm:optimal}
\end{theorem}

\begin{proof} 
	By individual rationality $F_l(\underline{R})R(0,0)$. Further, $F_l$ such that $F_l(R)=(0,0)$ for all $R$ is not optimal.    
	We note that optimal mechanism will always contain at least two elements.  
	If $Rn(F_l)=\{(t^1,q^1)\mid t^1>0\}$, i.e., $\#Rn(F_l)=1$, then by individual rationality we can set $F(\underline{R})=(0,q^{0})$
	where $(0,q^{0})\underline{I}(t^1,q^1)$. Since $\mu$ is continuous such readjustment does not change the expected revenue. If $Rn(F_l)=\{(0,q^1)\}$, then let $F_l(\overline{R})=(t^1,1)$ where $(0,q^1)\overline{R} (t^1,1)$. Again by continuity of $\mu$ such readjustment does not change the revenue. Now consider the following optimization problem.

	\begin{equation}
	\max_{t^{k},q^{k},R^{k}, k=0,\ldots,l-1} \sum_{k=0}^{l-1}~~t^{k}\mu([R^{k},R^{k+1}])
	\label{eqn_optimization_problem}
	\end{equation}

	\noindent s.t.
	$R^{0}= \underline{R}, R^{l}=\overline{R}$, $(t^0,q^0)R^{0}(0,0)$, 
	$(t^{k},q^{k})I^{k}(t^{k-1},q^{k-1}), k=1\ldots,l-1$
	$R^{k-1}\precsim R^{k}$  for $k=1,\ldots,l$. 
	$0\leq t^{0}, t^{k-1}\leq t^{k}\leq  \overline{T}$ and $0\leq q^{k-1}\leq q^{k}\leq 1$ for $k=1,\ldots,l-1$,  
	\label{lemma:sp_ir_1}
	
	\medskip

	\noindent where $\overline{T}$ is such that $(0,0)\overline{R}(\overline{T},1)$.  We claim a solution to this maximization problem exists and any solution can be ascribed to a strategy-proof and individually rational mechanism and the mechanism is an optimal mechanism $F_{l}^{*}$. To begin with note that the solution can have at most $l$ distinct bundles, at most $l-1$ preferences other than $\underline{R}$ and $\overline{R}$.
	The preferences $R^{k}$s are the analogs of the special preferences in a strategy-proof mechanism. If $l=3$, then we recall the special preferences $R^1$ and $R^2$ in Figure $1$. However there is a caveat,  
	we note that the manner in which the optimization problem is written it does not tell us anything about a mechanism. 
	We try to clarify this point. Suppose $l=5$, all are distinct and let $R^1=\ldots=R^{l-1}=R^*$. Then all the $5$ bundles fall on the same indifference curve of the preference $R^*$. Such a situation does not correspond to a strategy-proof mechanism. By the single-crossing property a strategy-proof mechanism requires $F(R^{*})\leq F(R)$ for all $R^*\precsim R$ and         
	$F(R)\leq F(R^{*})$ for all $R\precsim R^*$. Since all the bundles lie on just one indifference curve of $R^*$ 
	at the most three bundles can be assigned as the range of a strategy-proof mechanism. Thus two bundles will be left. 
	However this does not create any problem, we can just throw away two bundles and get a three bundle strategy-proof mechanism. 
	That is, we have a mechanism where $F(R)=(t^0,q^0)$ if $R\prec R^*, F(R^*)=(t^1,q^1), F(R)=(t^4,q^4)$ for $R^*\prec R$. 
	Then $R^*$ is the special preference for this three bundle mechanism. Since $\mu$ is continuous expected revenue from this mechanism is $t^0 \mu  ([\underline{R}, R^*[) + t^4 \mu ([R^*, \overline{R} ])$. This mechanism is in $\mathcal{F}_{5}$. 
	
	Thus at most for $l$ distinct bundles the objective function in the optimization problem is the expected revenue for a mechanism $F_l$ described by the following procedure: $F_l(R^{0})=(t^0,q^{0})$, $F_l(R^{1})=(t^{1},q^{1}),\ldots, F_l(\overline{R})=(t^{l-1},q^{l-1})$. 
	Further, $F_l(R)=(t^0,q^{0})$ for all $R\in [\underline{R},R^{1}[$,  $F_l(R)=(t^{1},q^{1})$ for all $R\in [R^{1},R^{2}[$, $F_l(R)=(t^{2},q^{2})$ for all $R\in [R^{2},R^{3}[\ldots$
	and $F_l(R)=(t^{l-1},q^{l-1})$ for all $R\in [R^{l-1},\overline{R}]$.
	By the single-crossing property $F_l$ is strategy-proof and individually rational. 
	In other words, if we solve the maximization problem in Equation (\ref{eqn_optimization_problem}) subject to the constraints, then the maximum obtained can indeed be ascribed to a strategy-proof and individually rational mechanism. Since a strategy-proof mechanism with range $l$ or less is completely determined by the special preference $R^1,\ldots, R^{l-1}$ the sum in Equation (\ref{eqn_optimization_problem}) corresponds is the expected revenue from the mechanism. 
	Thus instead of searching in the space of strategy-proof mechanisms we search in the $2l$ dimensional space of bundles, i.e., $t$s and $q$, and $l-1$ dimensional space of preferences, i.e., $l-1$ special preferences, for an optimal mechanism and that is enough, thanks to the single-crossing property of the CM preferences.           
	Thus to compete the proof of the theorem we
	prove the following lemma. 
	
	\begin{lemma}[{\bf Optimal}]\rm Let $\mu$ be non-trivial and continuous. The optimization problem in Equation (\ref{eqn_optimization_problem}) has a solution.
		\label{lemma:soln_exists}
	\end{lemma}
	
	\begin{proof} See Appendix $2$.

	\end{proof}

\end{proof}

\begin{remark}\rm 
	Note that  $M:\{(S,R)\mid S\precsim R, S, R\in [\underline{R},\overline{R}] \}\rightarrow [0,1]$ given by 
	$M(S,R)=\mu ([S,R])$ is a well defined function because $\mu$ is well defined for every $[S,R]$. The latter is because we have the Borel sigma algebra on $[\underline{R},\overline{R}]$ due to the order topology. The domain of $M$ is a bounded half space in $[\underline{R},\overline{R}]\times [\underline{R},\overline{R}]$. Since the preferences do not come from a vector space, the objective function is not a linear functional.
		We have not assumed any topology on the set of mechanisms. We look for a maximum of a function where the constraint set has a certain 
	geometric structure. These structure can be ascribed to a strategy-proof mechanism by the single-crossing property. Hence we do not need any topological assumption on the space of mechanisms to find an optimal mechanisms. As a proof of Lemma 
	\ref{lemma:soln_exists} we show that the objective function is continuous and the constraint set is compact.    
	Also note that $E(F_{l}^{*})\leq \overline{T}$. The constraint set in Theorem \ref{thm:optimal} does not include individual rationality constraints. This is because by the single-crossing property it is guaranteed by $(t^0,q^0)R^0 (0,0), (t^k,q^k)I^k (t^{k-1},q^{k-1})$ and $R^{k-1}\precsim R^k$. They can be included without any further complication since the proof that establishes Lemma \ref{lemma:preference_conv} also establishes that 
	if $\lim_{n\rightarrow \infty}(t^n,q^n)=(t,q), R^n\rightarrow R$, then 
	$(t^n,q^n)R^n(0,0)$ implies $(t,q)R(0,0)$. To show the existence of an optimal mechanism the non inclusion of these constraints do not matter.

\end{remark}

\subsubsection{\bf Optimal Mechanism with quasilinear preferences}  
\label{subsec:qlin}            

The preferences in Example \ref{ex:qlin} are given by $\theta q-t$. Let the seller receives $t$. 
Let $\theta\in [\underline{\theta},\overline{\theta}]$. Also let $\underline{\theta}>0$ and $\overline{\theta}<\infty$. 
We can identify the preference $\theta q-t$ by the pay-off relevant parameter $\theta$.
That is, the order on the preferences follow the natural order in $[\underline{\theta},\overline{\theta}]$.         
Further, let $\theta$ follow the distribution $\Gamma$, and $\gamma$ is the continuous density. In the next proposition we show that 
if $\Gamma$ satisfies increasing hazard rate, then the optimal mechanism for every 
$l\geq 2$
is deterministic.
A deterministic mechanism has two bundles in its range, one is of the kind 
$(t,1), t>0$ and the other is $(0,0)$. 
In Proposition \ref{prop:qlin1} we assume that $\Gamma$ admits increasing hazard rate.    
We keep the proof of Proposition \ref{prop:qlin1}
in the main text since the proof brings out the interaction between increasing hazard rate and slopes 
of indifference curves which entails the deterministic mechanism. Proposition \ref{prop:qlin1} demonstrates that we can study  this very important result in mechanism design  theory by using our framework as well. Let SP stand for strategy-proof and IR  for individually rational.

\begin{prop}\rm \label{prop:qlin1}   Let $\mathcal{F}=\{F\mid F:[\underline{\theta},\overline{\theta}]\rightarrow \mathbb{Z}, ~\text{SP, IR
		has finite}~Rn(F)\}$. 
	Let $\frac{\gamma(\theta')}{1-\Gamma(\theta')}\leq \frac{\gamma(\theta'')}{1-\Gamma(\theta'')}$ if $\theta'<\theta''$. Then for all $l$, $F_{l}^{*}$ is a deterministic mechanism.
	\end{prop}          

\begin{proof} Without loss of generality we can assume that $Rn(F)$ has 	
	bundles $(0,0)$ and $(t^{l},1)$. Since $R^{l}=\overline{R}$, by individual rationality at $F(R^{l-1})$
	consider $(t, 1)$ such that $\theta^{l-1}q^{l}-t^{l}=\theta^{l-1}-t$. 
	Then for every $R\in ]R^{l-1},R^{l}]$, 
	set $F(R)=(t,1)$, call this payment $t^{l}$. 
	Considering $(t^l,1)$ is without loss of generality since we are analyzing optimal mechanisms. 
	We fix $q$, and show that the first order conditions imply that the optimal mechanism must have at most two bundles in its range. Consider the optimization problem so that that range can have at most $l$ bundles. By Lemma \ref{lemma:soln_exists} we know that optimum solution exists.    
	Let by way of contradiction  $t^{k-1}<t^{k}$ and $q^{k-1}<q^{k}, \theta^{k-1}<\theta^{k}$.
	Further, we do not have to consider a bundle $(t,q)$ in the range such that $0<t^1<\overline{T}, 0<q<1$ such that the bundle is allocated at only for one single $R$. This is because probability of $\{R\}$ occurring is zero.           
	
	Consider the Lagrange  $L= \sum_{k=0}^{l-1}~~t^{k}[\Gamma(\theta^{k+1})-\Gamma(\theta^{k})]$
	$-\lambda_{1}[\theta^{1}q^{1}-t^{1}]-\sum_{k=2}^{l-1} \lambda_{k} [\theta^{k}q^{k}-t^{k}-\theta^{k}q^{k-1}+t^{k-1}]$. Since $F$ is optimal $t^{0}=0=q^{0},q^{l-1}=1$. Also, $\Gamma(\theta^{l})=\Gamma(\overline{\theta})=1$ and $\Gamma(\theta^{0})=\Gamma(\underline{\theta})=0$. The first order conditions with respect to $\theta^{1}$ is:         
	
	\begin{equation}
	-t^{1}\gamma(\theta^{1})-\lambda_1 q^{1}=0\\
	\implies \frac{\gamma(\theta^{1})}{-\lambda_{1}}=\frac{q^{1}}{t^{1}}. 
	\label{eqn:theta_1_myer}
	\end{equation} 
	\noindent The first order condition with respect to $\theta^k$, $k=1,\ldots,l-1$ is:   
	\begin{equation}
	t^{k-1}\gamma(\theta^{k})-t^{k}\gamma(\theta^{k})-\lambda_{k}[q^{k}-q^{k-1}]=0
	\implies \frac{\gamma(\theta^{k})}{-\lambda_{k}}=\frac{q^{k}-q^{k-1}}{t^{k}-t^{k-1}}.
	\label{eqn:theta-k_myer}
	\end{equation} 
	\noindent The first order condition with respect to $t^{k}$ for $k=1,\ldots,l-2$ is
	\begin{equation}
	[\Gamma(\theta^{k+1})-\Gamma(\theta^{k})]+\lambda_{k}-\lambda_{k+1}=0.
	\label{eqn:t_k_myer}
	\end{equation}
	\noindent The first order condition with respect to $t^{l-1}$ is
	\begin{equation}
	[1-\Gamma(\theta^{l-1})]+\lambda_{l-1}=0.
	\label{eqn:t_l-1-kmyer}
	\end{equation}

	\noindent from equations (\ref{eqn:t_k_myer}) and (\ref{eqn:t_l-1-kmyer})  $\lambda_{k}=-[1-\Gamma(\theta^{k})]$ for $k=1,\ldots, l-1$. Then from equations (\ref{eqn:theta_1_myer}) and (\ref{eqn:theta-k_myer}) we obtain 
	$\frac{\gamma(\theta^{1})}{1-\Gamma(\theta^{1})}=\frac{q^{1}-0}{t^{1}-0}$, and 
	$\frac{\gamma(\theta^{k})}{1-\Gamma(\theta^{k})}=\frac{q^{k}-q^{k-1}}{t^{k}-t^{k-1}}=\frac{1}{\theta^{k}}$ for
	$k=2,\ldots,l-1$. Thus $\theta^{k-1}<\theta^{k}$ implies 
	$\frac{q^{k}-q^{k-1}}{t^{k}-t^{k-1}}<\frac{q^{k-1}-q^{k-2}}{t^{k-1}-t^{k-2}}$. 
	Since $\frac{\gamma(\theta')}{1-\Gamma(\theta')}\leq \frac{\gamma(\theta'')}{1-\Gamma(\theta'')}$ if $\theta'<\theta''$ we reach at a contradiction. Thus, the proof follows.

\end{proof}

\noindent Since $\theta^{k}q^k-t^k=\theta^{k}q^{k-1}-t^{k-1}$, $\frac{q^{k}-q^{k-1}}{t^{k}-t^{k-1}}$ is the slope of the indifference curves of preferences that correspond to $\theta^{k}$. 
The optimal mechanism requires this slope to be equal to the hazard rate at $\theta^{k}$, i.e., 
$\frac{\gamma(\theta^{k})}{1-\Gamma(\theta^{k})}=\frac{q^{k}-q^{k-1}}{t^{k}-t^{k-1}}\equiv \frac{\text{increase in winnig probability}}{\text{increase in payment}}$

$\equiv\text{relative, i.e., relative to increase in payment, increase in winning probability}$.

\noindent The hazard rate at $\theta^{k}$ measures the conditional probability that the buyer's type fails to be below $\theta^{k}$. Thus increasing hazard rate in a sense implies that it is more likely that the buyer's type is high and not low. The single-crossing property implies that the
relative increase in the winning probability decreases as types increase. 
The proof of Proposition \ref{prop:qlin1} shows that the equality between the hazard rate and the relative increase in the winning probability holds only for a deterministic mechanism.       
The next proposition pins down the the optimal mechanism $F^{*}$.

\begin{cor}\rm Let $\frac{\gamma(\theta')}{1-\Gamma(\theta')}\leq \frac{\gamma(\theta'')}{1-\Gamma(\theta'')}$ if $\theta'<\theta''$. Then $F^{*}$ exists. Further, $F^{*}$ is defined as follows:    		
	
	$$t^{*}(\theta) = \begin{cases}
	\theta^{*}, & \text{if $\theta>\theta^{*}$;}\\
	$0$, & \text{if $\theta\leq \theta^{*}$.}
	\end{cases}$$

	$$q^{*}(\theta) = \begin{cases}
	1, & \text{if $\theta>\theta^{*}$;}\\
	$0$, & \text{if $\theta\leq \theta^{*}$.}
	\end{cases}$$

	\label{cor:qlin_opt}

	\noindent Where $t^{*}, \theta^*$ solve $\max _{t,\theta}t[1-\Gamma(\theta)]$ subject to $\theta-t=0$.  
\end{cor}

\begin{proof} From Proposition \ref{prop:qlin1} it us enough to look for an optimal mechanism within the class of deterministic mechanisms. Then we note that a solution to the optimal problem exists.{\footnote{The constraints satisfy the non degenerate constraint qualification, NDCQ, condition as stated in Theorem $18.5$ in \citep{Simon}.}}     
	
\end{proof}

\noindent In the next remark we argue that for distributions without monotone hazard rate as well the optimal mechanism is deterministic if the preferences are linear.    

\begin{remark}\rm The first order equations in Proposition \ref{prop:qlin1}  	
require the $k^{\text{th}}$ hazard rate to be equal to the slope of an indifference curve of the $k^{\text{th}}$ preference $\theta^k$ and the deterministic mechanism satisfies these conditions irrespective of whether the hazard rate is monotone or not. In fact we can see that only the deterministic mechanism can be a solution to this problem. Suppose $l=4$, and set $t^0=q^0=0$ by optimality. Let $(t^{1*},t^{2*},t^{3*}, q^{1*},q^{2*}, q^{3*}=1, \theta^{1*}, \theta^{2*}, \theta^{3*})$ be a solution to the optimization problem.      
Then consider 
$\alpha=\max_{t^1, q^1, \theta^1} t_1(\Gamma(\theta^{2*})-\Gamma(\theta^1))+t^{2^*}(\Gamma(\theta^{3*})-\Gamma(\theta^{2*}))+t^{3*}(1-\Gamma(\theta^{3*}))$ 
such that $\theta^1q^1-t^1=0, \theta^{2*}q^1-t^1=\theta^{2*}q^{2*}-t^{2*}, 
\theta^{2*}q^{2*}-t^{2*}=\theta^{3*}-t^{3*}$. That is, we fix all the indifference curves from $\theta^{2*}$ onward for the optimal solution $(0, t^{1*},t^{2*},t^{3*},0, q^{1*},q^{2*}, q^{3*}=1, \theta^{1*}, \theta^{2*}, \theta^{3*})$. Since $t^{1*}, \theta^{1*}, q^{1*}$ is a feasible solution for the $\alpha$ problem 
we have $\alpha=  t^{1*}(\Gamma(\theta^{2*})-\Gamma(\theta^{1*}))+t^{2*}(\Gamma(\theta^{3*})-\Gamma(\theta^{2*}))+t^{3*}(1-\Gamma(\theta^{3*}))$. Now from $\theta^1q^1-t^1=0$ we obtain
$\theta^1=\frac{t^1}{q^1}$. Thus, $\frac{\gamma(\theta^1)}{1-\Gamma(\theta^1)}=\frac{q^1}{t^1}=\frac{1}{\theta^1}$.
From $\theta^{2*}q^1-t^1=\theta^{2*}q^{2*}-t^{2*}$ we obtain
$t^1=\theta^{2*}q^1-\theta^{2*}q^{2*}+t^{2*}$.
Thus, by using all the relevant constraints we obtain the optimization problem 
$\max_{q^1,\theta^1}(\theta^{2*}q^1-\theta^{2*}q^{2*}+t^{2*})(\theta^{2*}-\frac{1-\Gamma(\theta^1)}{\gamma(\theta^1)})$. If $0<q^1<1$, then the first order condition is $(\theta^{2*}-\frac{1-\Gamma(\theta^1)}{\gamma(\theta^1)})=0$. Thus $\theta^1=\theta^{2*}$. Thus the mechanism must have at most three distinct elements. By extending this argument we see that the mechanism has to be deterministic.  In Corollary $3.1$ in \citep{Goswami2} it is shown that  if an optimal mechanism exists among the finite range mechanisms, then it must be optimal in general.   
\end{remark}	

\noindent In the next example we study the class preferences 
that satisfy positive income difference mentioned in Example \ref{ex:index}.

\begin{example}\rm Let the utility function of the buyer be given by $\theta q-t^2$ where $\theta \in [1,2]$. Since the utility function has power, in this example we use subscripts to as labels for bundles. In Example \ref{ex:index} we have seen that these preferences satisfy positive income difference. Let  $t \in [0, \sqrt{2}]$, where  $2.1-\sqrt{2}^2=0$. 
	Let $ \Gamma(\theta) = \theta - 1 $, $ l = 3$. Now, the objective function is 
	$[t_0 (\theta_1 - \theta_0) ] + [t_1 ((\theta_2 - 1) - (\theta_1 - 1))] + [t_2 (1 - (\theta_2 - 1))]$, $\theta_0=1$.  
	By optimality let $ t_0 = 0 $, $q_{0} = 0 $ $ q_{l-1} = 1$. Further, we have   $ t_0 \leq t_1 \leq t_2$, and $ q_0 \leq q_1 \leq q_2$. Since $t_0=0$ the objective function can be written as $t_1 (\theta_2 - \theta_1) + t_2(2-\theta_2)$. 
	By individual rationality, $t_1^2 - \theta_1 q_1=0$ and $t_2^2 - \theta_2 \leq 0$ since $q_2=1$.
	By optimality $q_2=1$. 
	The optimal mechanism is a solution to  
	the Lagrange: $L=t_1 (\theta_2 - \theta_1) + t_2(2-\theta_2)-\lambda_1[(\theta_2-t_{2}^2)-(\theta_2 q_1-t_1^2)]
	-\lambda_2(t_1^2- \theta_1 q_1)-\mu_2 (t_{2}^2-\theta_2)- \mu_3(\theta_1-\theta_2)-\mu_4(t_1-t_2)+\mu_5 t_1+\mu_6 t_2-\mu_7(t_1-\sqrt{2})- \mu_8(t_2-\sqrt{2})-\mu_9(\theta_1-2)-\mu_{10}(\theta_2-2)-\mu_{11}(-\theta_1+1)-\mu_{12}(-\theta_2+1)-\mu_{13}(q_1-1)+\mu_{14}q_1$. The First order conditions with respect to
	\begin{itemize}
		\item $q_1$ is $\lambda_1 \theta_2+ \lambda_2 \theta_1-\mu_{13}=0$   
		\item $\theta_1$ is $-t_1+\lambda_2q_1-\mu_3-\mu_9+\mu_{11}=0$
		\item $\theta_2$ is $t_1-t_2-\lambda_1 [1-q_1]+\mu_2+\mu_3-\mu_{10}+\mu_{12}=0$
		\item $t_1$ is $(\theta_2-\theta_1) -2\lambda_1 t_1-2\lambda_2t_1-\mu_4+\mu_5-\mu_7=0$
		\item $t_2$ is $(2-\theta_2)+2\lambda_1 t_2-2t_2\mu_2+\mu_4+\mu_6-\mu_8=0$		
	\end{itemize}
	
	\noindent If $\theta_1=\theta_2=1, t_1=0, t_2=1, q_1=0, \lambda_1=0=\lambda_2, \mu_2=\frac{1}{2}=\mu_{3}=\mu_{11}, \mu_{k}=0$ for all other $\mu_{k}$s, then they satisfy the first order conditions and satisfy all the constraints. For the second order sufficient conditions we note that the following inequality constraints bind: 
	$\theta_1- \theta_2=0$, $t_2^2-\theta_2=0$, $t_1=0$, 
	$-\theta_1+1=0$, $-\theta_2+1=0$, $q_1=0$. We call these constraints $g_1, \ldots, g_6$. 
	Now consider gradient of each $g_i$s with respect to $t_1, t_2, \theta_1, \theta_2, q_1$ respectively. We have $\nabla g_1=(0,0, 1, -1, 0)$
	$\nabla g_2=(0, 2t_2, 0, -1, 0)$, $\nabla g_3=(1, 0, 0, 0, 0)$, $\nabla g_4=(0,0, -1, 0,0)$, 
	$\nabla g_{5}=(0,0,0, -1,0 )$,
	$\nabla g_6=(0,0,0, 0, 1)$. The Jacobian matrix for the inequality matrix evaluated at the critical values is:

	$$	\begin{bmatrix}
		0 & 0 & 1 & -1 & 0 \\
		0 & 2 & 0 & -1 & 0 \\
		1 & 0 & 0 & 0 & 0 \\
		0 & 0 & -1 & 0 & 0 \\
		0 & 0 & 0 & -1 & 0 \\
		0 & 0 & 0 & 0 & 1 \\
	\end{bmatrix}$$

	\noindent This matrix has full rank which is $5$ and thus its null space consists of $(0,0,0,0,0)$. 
	Thus by Theorem $19.8$ the first order solution is a strict local maximizer. 
In fact the solution is a global maximizer. 
We show that the constraint maximization problem does not have maximize in interior $t_1$ and $q_1$. Let $(t_1^*, t_2^*, q_1^*, \theta_1^*, \theta_2^*)$ be a solution to the optimization problem such that $0<t_1^*<\sqrt{2}, 0<q_1^*<1$. Let  $\alpha=\max_{t_1, q_1, \theta_1} t_1(\theta_2^*-\theta_1)+t_2^*(2-\theta_2^*)$ such that $\theta_1q_1-t_1^{2}=0, \theta_2^*q_1-t_1^{2}=\theta_2^*-t_2^{*2}$.      
Then $\alpha= t_1^*(\theta_2^*-\theta_1^*)+t_2^*(2-\theta_2^*)$. This equality must hold because $(t_1^*, q_1^*, \theta_1^*)$ are feasible solution to the $\alpha$ constrained maximization problem. Now replacing for $\theta_1$ and $t_{1}^2$ the constraints in the 
$\alpha$ optimization problem we obtain $\max_{t_1, q_1, \theta_1} t_1(\theta_2^*-\frac{\theta_2^*q_1-\theta_2^*+t_{2}^{*2}}{q_{1}})+t_2^*(2-\theta_2^*)$. Then first order condition with respect to $t_1$ and $q_1$ are 
$\theta_2^*=0$ and $\frac{1}{q_1^2}[-\theta_2^*+t_{2}^{*2}]=0$. These conditions do not entail an interior solution.$\square${\footnote{Alternatively we could have first argued that that the optimal mechanism must be deterministic, and then we could have solved for the optimal deterministic mechanism. This approach is shorter but may not be considered very helpful in understanding the structure of the problem.}}           
 \label{ex:optimal_positive_difference}
\end{example}	

\noindent We wish to make a remark about the proof of the global maximum in Example \ref{ex:optimal_positive_difference}. 

\begin{remark}\rm If there were more than three possible distinct bundles 
	then also the trick to check if the optimal mechanism is deterministic applies. Suppose $l=4$. Then consider
		$\alpha=\max_{t_1, q_1, \theta_1} t_1(\theta_2^*-\theta_1)+t_2^*(\theta_3^*-\theta_2^*)+t_3^*(2-\theta_3^*)$ such that $\theta_1q_1-t_1^{2}=0, \theta_2^*q_1-t_1^{2}=\theta_2^*q_2^*-t_2^{*2}, 
	\theta_2^*q_2^*-t_2^{*2}=\theta_3^*-t_3^{*2}$. Thus, by the same argument as in Example \ref{ex:optimal_positive_difference} the $t_1$ cannot be interior, but then repeating the argument once again $t_2$ cannot be interior. Now suppose that the buyer's pay-off function is of the form $\theta \sqrt{q}-t^2, \theta>0$, where $\sqrt{q}$ is a probability weighing function, which puts almost $0$ weights to the probabilities close to $0$. In this case we have $\alpha=\max_{t_1, q_1, \theta_1} t_1(\theta_2^*-\theta_1)+t_2^*(\theta_3^*-\theta_2^*)+t_3^*(2-\theta_3^*)$ such that $\theta_1\sqrt{q_1}-t_1^{2}=0, \theta_2^*\sqrt{q_1}-t_1^{2}=\theta_2^*\sqrt{q_2^*}-t_2^{*2}, 
	\theta_2^*\sqrt{q_2^*}-t_2^{*2}=\theta_3^*-t_3^{*2}$. Again the same argument applies and shows that interior solution is not possible. However, things change if the utility of the buyer is $q\sqrt{\theta-t}, t\leq \theta$ i.e.,
	the buyer with valuation $\theta$ is risk averse with concave Bernoulli utility function facing the lottery $q, 1-q$ with consequences $\theta-t$ and $0$ respectively. Such risk averse pay-off function leads to non-monotonicity and thus do not belong to the CMS class. We study such risk averse preferences in \citep{Goswami1}. The observation that we wish to make here is that optimization program for non-monotone preferences maybe more complicated as we loose separability of preferences. We make a further observation about how our solution techniques are different from \citep{myer}. 
	Consider the preference $\theta q-t^2$. Now $U(\theta)=\max_{z \in [\underline{\theta}, \overline{\theta}]} \theta q(z)-t^2(z)$  is convex since it $U(\theta)$ is a maximum of a family affine functions in the true valuation $\theta$, and thus $U'(\theta)=q(\theta)$ almost everywhere. 
	Thus from strategy-proofness we have $t^{2}(\theta)=-\underline{\theta}q(\underline{\theta}) + t^{2}(\underline{\theta})+ q(\theta)\theta-\int_{\underline{\theta}}^{\theta}q(x)dx$. Thus, 
	$t(\theta)=\sqrt{-\underline{\theta}q(\underline{\theta})+t^{2}(\underline{\theta})+ q(\theta)\theta-\int_{\underline{\theta}}^{\theta}q(x)dx}$. 
	Now revenue equivalence does not hold. Consider an allocation rule such that 
	$q(\theta)=0 $ if $\theta<\theta^*$ and $q(\theta)=1 $ if $\theta^*\leq \theta$. 
	Let $t_1$ and $t_2$ be two distinct payment rules. Then, $t_1(\underline{\theta})-t_2(\underline{\theta})=\sqrt{t_{1}(\underline{\theta})^2}-\sqrt{t_{1}(\underline{\theta})^2}\neq t_1(\overline{\theta})-t_2(\overline{\theta})=\sqrt{t_{1}(\underline{\theta})^2+c}-\sqrt{t_{2}(\underline{\theta})^2+c}$, 
		where $c=\overline{\theta}-\int_{ \underline{\theta}  } ^{\overline{\theta}}q(x)dx=\overline{\theta}-[\overline{\theta}-\theta^*]=
		\theta^*$. However our point is not about revenue equivalence not holding, the point that we are trying to make is that computing expected revenue and computing the optimal mechanism by substituting the payment rule in seller's objective makes the computations harder.                
		The expected revenue is $ \int_{\underline{\theta}}^{\overline{\theta}} \sqrt{-\underline{\theta}q(\underline{\theta})+t^{2}(\underline{\theta})+ q(\theta)\theta-\int_{\underline{\theta}}^{\theta}q(x)dx}~ \gamma(\theta) d\theta$, where $\gamma(\theta)$ is the density function of the valuations. Thus, the finding optimal mechanism involves solving this integral equation subject to the individual rationality constraint
	 $t^2(\underline{\theta})\leq 0$ and $q$ is monotone. If the preferences are of the form $\theta \sqrt{q}-t^2  $, then the integral equation is $ \int_{\underline{\theta}}^{\overline{\theta}} \sqrt{-\underline{\theta}\sqrt{q(\underline{\theta})}+t^{2}(\underline{\theta})+ \theta\sqrt{q(\theta)}-\int_{\underline{\theta}}^{\theta}\sqrt{q(x)}dx}~ \gamma(\theta) d\theta$.

\end{remark}

\noindent In order to highlight the importance of the order on the preferences we make a few observations about the model with 
multidimensional parameters in Example \ref{example:two-parameter} next.

\begin{example}\rm 
Let $\theta \in [1,2], \alpha \in [\frac{1}{2},1]$. The preferences in $U$ are `smaller' than the preferences in $V$ according to the order on $U\cup V$ entailed by the single-crossing condition. Now consider a mechanism in $\mathcal{F}_{l}$. Instead of writing the mechanism in terms of preferences we can also write it in terms of the parameters $\theta$ and $\alpha$. Consider a strategy-proof mechanism that uses special preferences from both $U$ and $V$. There are two cases. One of the cases is when $2\sqrt{q}-t$ is a special preference, i.e., the common preference in $U$ and $V$. In this case the geometry entailed in terms of special preferences is clear because $2\sqrt{q}-t$ is the largest preference in $U$ according the order on preferences due to the single-crossing property. If the common preference is not a special preference, then let $\theta'\sqrt{q}-t$ be the largest special preference from $U$. Since the number of alternatives in the range is       
finite the largest special preference is well defined. Then the special preference after $\theta'\sqrt{q}-t$ is of the form $2\sqrt{q}-\alpha t$. To write the mechanism consistently let $\beta=\frac{1}{\alpha}$. 
Thus we write 
$\theta^1 \sqrt{q ^1}-\frac{1}{\beta^1} t^1=\theta^1 \sqrt{q ^0}- \frac{1}{\beta^1} t^0, \ldots ,\theta^{k+1}\sqrt{q ^{k}}-\frac{1}{\beta^{k+1}} t^k=\theta^{k+1}\sqrt{q ^{k+1}}-\frac{1}{\beta^{k+1}} t^{k+1}$ and search for values of $\theta$ and $\beta$,
$(\theta^k, \beta^k )\in U'=\{(\theta, \beta) \mid \beta =1, \theta \in [1,2]\}   \cup V'=\{(\theta, \alpha)\mid \theta=2, \beta \in [1, 2]\}$. Let $\Gamma$ be then joint distribution on $U'\cup V'$ and consider it as a subspace of the Euclidean space $\Re^{2}$, and thus we can consider the Borel sigma algebra of this subspace. 
Then the set $[\theta', \theta'']\times \{1\}$ and $\{2\}\times [\beta', \beta'']$ are well defined Borel sets. 
Then consider  joint distribution $\Gamma$ such that $\Gamma (\theta', 1)< \Gamma (\theta'', 1)$ for $\theta '<\theta''$, $\Gamma (2, \beta') < \Gamma (2, \beta'')$ for $\beta'<\beta''$, and for any $\theta \in U'$ and $\beta \in V'$, $\Gamma (\theta, 1) < \Gamma (2, \beta)$.$\square$ 

\end{example}

\noindent In the next section we make a remark about how to extend these ideas to the situation when $q$ denotes quality.  

\subsubsection{\bf Extension to the situation where $q$ refers to quality}
\label{sec:quality}
Let $q$ denote quality, and $q\in [\underline{q},\overline{q}]$. 
That is,  let $[\underline{q},\overline{q}]$ be a linear continuum and endowed with the corresponding order topology. With the product topology on $\Re_{+} \times [\underline{q},\overline{q}] $, CMS preferences are well defined on $[\underline{q},\overline{q}]\times \Re_{+}$.          
 Let us further assume that $[\underline{q},\overline{q}]$ is homeomorphic to a closed interval in the real line, say $[0,1]$, i.e., $[\underline{q},\overline{q}]$ is an one dimensional topological manifold. 
 As an example we  may consider $q$ to be an index of quality. In such a situation a representation a preference by a formula, for example $\theta q-t$, may not be well defined since $q$ is not a real number.
 The point is that, in case of discrete quality say high and low, high can be labeled as $1$ and low can be labeled as $0$. Here $1-0$ or $1+0$ may not have any meaning. 
 If quality lies in $[\underline{q},\overline{q}]$, then a homeomorphism only guarantees that $q$ is drawn from an one dimensional manifold, and this does not imply that we can add or multiply $q$s like real numbers. 
 If quality comes from a continuum  $[\underline{q}, \overline{q}]$, then considering it as coming from $[0,1]$ is just a relabeling of quality by using real numbers from $[0,1]$. 
 Considering quality in $[0,1]$ as being drawn from an Euclidean space so that we can use standard calculus is an assumption, and such an assumption need be valid  in all circumstances.
 As a result a continuous preference in $\Re_{+} \times [\underline{q},\overline{q}] $ need not be expressed in terms of a formula.    
 Certainly, in certain situations $q$ may be considered to be cardinal, and thus a formula may be well defined and also available; a formula need not be available since in general a continuous function need not have a representation in terms of a known closed form formula. Our general analysis of strategy-proofness and existence proof of optimal mechanism do not depend on the formulae of the preferences.
 Certainly, finding an optimal mechanism without using calculus is not an easy task. 
 However, given that computations on abstract topological spaces is active area of research, it will be interesting to see how the optimization program that we have proposed can be solved, analytical or approximate, using the computational techniques on toplogical spaces. 
 
\citep{Mishra1} consider a model where $q$ comes from an ordered finite set. If we consider restrictions of CMS preferences to the situation with $q$s coming from a finite set, then our characterization of strategy-proof mechanisms hold. Guaranteeing existence of an optimal mechanism is not straightforward. In fact \citep{Mishra1} assume that an optimal mechanism exists and then state what it is. The reason guaranteeing the existence of an optimal mechanism is hard because the constraint set is not compact anymore. To see the problem consider $q^1=0<q^2<q^3=1$ be three fixed qualities of a good that can be sold, i.e., it is a supply side constraint. The optimal mechanism design question is now about how to price these qualities optimally. Further, suppose that an external regulation requires that the firm must have qualities of all three goods available. Then a strategy-proof mechanism requires $t^0<t^1<t^2$. Just to explain the point suppose $t^0=0$.
Then note $\{(t^1,t^2)\mid t^1-t^2<0\}$ is an open set, and thus guaranteeing existence of an optimal mechanism is difficult. However, as we saw earlier from the geometry of strategy-proof mechanisms we can compute the expected revenue easily: $t^1\mu([R^1,R^2])+t^2\mu( [R^2, \overline{R}])$, where $(t^0, \underline{q})I^{1}(t^1,0), (t^1,q^1)I^{2}(t^2,1)$ which maybe computed with the help of modern computational techniques and maybe large enough iterations will give an approximate value of the maximum expected revenue. By imposing more restrictions, for example lower bound on the price of the highest quality, we can `bound' the geometry of strategy-proof mechanisms and make computations may be easier on $[\underline{R},\overline{R}]$.      
In the next section we make a remark about the implications for strategy-proofness if we extend the domain.     

\subsubsection{Extending a domain beyond singe-crossing} 
\label{sec:union_single_crossing}     
In this section we make a remark about optimal mechanisms when the domain of $F$ is an union of two single-crossing domains. We assume $Rn(F)$ to be finite. We consider the  domain $D=\{u((t,q); \theta, \frac{1}{2})=\theta\sqrt{q}-\frac{1}{2}t\mid \theta\in [1,5]\}\cup \{u((t,q); \theta, \frac{1}{3})=\theta q^{\frac{1}{3}} - \frac{1}{3}t\mid \theta \in [1,5]\}$.
We can represent $D$ parametrically as $\{(\theta, \alpha)\mid \theta\in [1,5], \alpha\in \{\frac{1}{2}, \frac{1}{3}\}\}$.
We show that
$D$ is not a single-crossing domain. 
We note that slope of an indifference curve of a preference in $D$ is given by
either $\frac{\sqrt{q}}{\theta_1}$ or $\frac{q^{\frac{2}{3}}}{\theta_2}$. Consider $\theta_1=2$ and 
$\frac{\sqrt{  \frac{1}{2}   }  }{2}=\frac{  \frac{1}{2}   ^{\frac{2}{3}}}{\theta_2}$.
Thus $\theta_2=2[\frac{1}{2}]^{\frac{1}{6}}$ which lies in the interval $[1,5]$. This means that  -
that the indifference curves of the preferences denoted by the parameter vectors $(2, \frac{1}{2})$ and 
$(2[\frac{1}{2}]^{\frac{1}{6}}, \frac{1}{3})$ are tangent to each other at $q=\frac{1}{2}$. 
That is the equality $  \frac{\sqrt{q}} {2} = \frac{q^{\frac{2}{3}}}{2[\frac{1}{2}]^{\frac{1}{6}}}$ holds at $q=\frac{1}{2}$. Thus $  \frac{\sqrt{q}} {2} < \frac{q^{\frac{2}{3}}}{2[\frac{1}{2}]^{\frac{1}{6}}}$ 
if and only if $[\frac{1}{2}]^{\frac{1}{6}}<q^{\frac{1}{6}}$ and  
$  \frac{\sqrt{q}} {2} > \frac{q^{\frac{2}{3}}}{2[\frac{1}{2}]^{\frac{1}{6}}}$ 
if and only if $[\frac{1}{2}]^{\frac{1}{6}}>q^{\frac{1}{6}}$. Since $q^{\frac{1}{6}}$ is an increasing function the equality holds only at $q=\frac{1}{2}$. Thus these two preferences are not the same. In fact $(2,\frac{1}{2})$ is a Maskin Monotonic Transformation of $( 2[\frac{1}{2}]^{\frac{1}{6}}, \frac{1}{3})$ through $(t,q)$, where $q=\frac{1}{2}$. For the sake of simplicity we call a subdomain in $D$ to be {\bf a single-crossing slice (in short we call it a `slice')}. 
For example $\{\theta\sqrt{q}-\frac{1}{2}t\mid \theta \in [1,5]\}$ is slice 1
and $\{\theta q^{\frac{1}{3}}-\frac{1}{3}t\mid \theta \in [1,5]\}$ is slice $2$. Compared with slice $2$, in slice $1$ overweighting of probability is smaller and the disutility of from paying for the good is larger.      
Assume that $\theta$ and $\alpha$ are independently distributed. Let Prob$( \{\frac{1}{2}\})=$ Prob $( \{\frac{1}{3}\})=\frac{1}{2}$. 
Let $\theta$ be distributed according to $\Gamma$ with differentiable density $\gamma$ and has increasing hazard rate.  
The optimal mechanisms for the slices in $D$ are deterministic. 
From the first order conditions we obtain  $\frac{\gamma( \theta^{k} ) }  { [1-\Gamma(\theta^k)] }
=\frac{   [(q^k)^{\alpha}-(q^{k-1})^{\alpha}] }        { \alpha(  t^k-t^{k-1})  }=\frac{1}{\theta_k}, \alpha \in \{\frac{1}{2},\frac{1}{3}\}$. By arguing as before, the optimal mechanisms are deterministic. 
Now let $\Gamma(\theta)=\frac{\theta-1}{4}$. Then $\gamma(\theta)=\frac{1}{4}$. 
However note that now the optimal payment in the first slice of $D$ is  $t_1=2\theta_1$ and for the second slice it is $t_2=3\theta_2$. 
Now $\arg\max_{\theta_1\in [1,5]}2\theta_1[1-\frac{\theta_1-1}{4}]=\arg\max_{\theta_2\in [1,5]}3\theta_2[1-\frac{\theta_2-1}{4}]$. At the optimum $\theta^*=2.5$ in both the slices. 

Now consider optimal mechanism in $D$. If the mechanism is strategy-proof and individually rational in $D$, then it is strategy-proof and individually rational in each slice. Let $E[F_l]$ denote the expected revenue from a strategy-proof mechanism defined on $D$ with maximum $l$ distinct elements. Then, let $E[F_{l_1}]$ and $E[F_{l_2}]$ be the optimal expected revenue from slice $1$ and $2$ respectively. We note that $E[F_{l_1}] < E[F_{l_2}]$.  Since $E[F_l]=Prob(\{\frac{1}{2}\})E[F_{l_1}]+Prob(\{\frac{1}{3}\})E[F_{l_2}]=\frac {1}{2}E[F_{l_1}]+\frac {1}{2}E[F_{l_2}]$, $E[F_{l_1}]< \frac{1}{2}E[F_{l_1}]+\frac {1}{2}E[F_{l_2}] <   E[F_{l_2}]$.  
Since $2\theta^*=t_1<t_2=3\theta^*$, $u(F_l(2.5, \frac{1}{2}); 2.5, \frac{1}{3})>u(F_l(2.5, \frac{1}{3}); 2.5, \frac{1}{3})$. This violates strategy-proofness. That is, the optimal mechanisms in each individual slice cannot be made to a strategy-proof mechanism for $D$. Now suppose there is 
 another mechanism. However, whatever be that mechanism the expected revenue from that mechanisms cannot be more than $E[F_{l_2}]$. If $\frac{1}{2}E[F_{l_1}^*]$ and $\frac{1}{2}E[F_{l_2^*}]$ are the expected revenues from each slice of the optimal mechanism, then
 $\frac{1}{2}E[F_{l_1}^*]+\frac{1}{2}E[F_{l_2}^*]\leq \frac{1}{2}E[F_{l_1}]+\frac{1}{2}E[F_{l_2}]<E[F_{l_2}]$.
  Suppose the auctioneer has a target revenue in mind, for the sake of argument let it be equal to the expected revenue from the optimal mechanism in slice $2$. Then there is no optimal mechanism that ensures the targeted revenue.  Since $\alpha$ occurs with equal probability the average expected revenue is less than the optimal revenue from slice $2$ alone. In the absence of preferences from slice 1 the auctioneer could have earned the targeted expected revenue.    
Thus a larger behavioral spread, which in this case is $\{\frac{1}{2},\frac{1}{3}\}$ need not be better for the auctioneer.

\subsubsection{{\bf An $n-$ Buyer Environment}}
\label{sec:n_buyer} 
Let $R=(R_{1},\ldots, R_{n}) \in \Pi_{i=1}^{n}[\underline{R},\overline{R}]$
be a profile of preferences. Also let for any $i$, $R_{-i}=(R_{1},\ldots,R_{i-1},R_{i+1}, \ldots, R_n)$ i.e., $R_{-i}$ is the profile of preference with buyer $i$ removed.
Consider the joint probability space $(\times_{i=1}^{n}[\underline{R},\overline{R}], \times_{i=1}^{n}{\cal B},\times_{i=1}^{n}\mu_i)$, where $\times_{i=1}^{n}{\cal B}$ denotes the product sigma-algebra, and $\times_{i=1}^{n}\mu_i$ denotes a probability measure on the product sigma algebra such that for all $i$ and for all $R_{-i}$, $\mu_{i}(\cdot, R_{-i})$ is a probability measure on $([\underline{R},\overline{R}], {\cal B})$. 
Let $F:\Pi_{i=1}^{n}[\underline{R},\overline{R}]\rightarrow \mathbb{Z}^{n}$ be  a $n-$ buyer mechanism.  
Then $F_{i}=(T_i,Q_i)$ is the $i^{th}$ component of $F$ and $F_{i}:\times_{i=1}^{n}[\underline{R},\overline{R}]\rightarrow \mathbb{Z}$ refers to allocation of buyer $i$. 
Also for all profile $R$, $\sum_{i=1}^{n}Q_{i}(R)\leq 1$. 
We write payment and probability of win in capital letters to distinguish the 
$1-$buyer environment from the $n-$buyer environment.  
Next we define strategy-proof and individually rational mechanism. 

\begin{defn}\rm An $n-$buyer mechanism $F$ is {\bf strategy-proof} for all $i$ and all $R_i, R_{i}'$, for all $R_{-i}$ , $(T_{i}(R_i, R_{-i}), Q_{i}(R_i, R_{-i}))R_i (T_{i}(R_i', R_{-i}), Q_{i}(R_i', R_{-i}))$. The mechanism is {\bf individually-rational} if for $i$, $(T_{i}(R_i, R_{-i}), Q_{i}(R_i, R_{-i}))R_i (0,0)$.   
\end{defn}

\noindent A mechanism is strategy-proof if irrespective of other buyers' announcements, a buyer has no incentive to misreport her own preference. Let $Rn(F_i(\cdot, R_{-i}))$ be the range of the mechanism $F$ when $n-1$ buyers' preferences are fixed. This is, the set of allocations that buyer obtains when 
other buyers' preferences are fixed at $R_{-i}$.  
Often it is also called the option set of buyer $i$ at $R_{-i}$. We assume that the option set of each buyer $i$ at each $R_{-i}$ is finite. Then we have the following result immediately from the one buyer environment.

\begin{theorem}\rm Let $F$ be a $n-$buyer mechanism. Let $Rn(F_i(\cdot, R_{-i}))$ be finite for every $i$ and every $R_{-i}$. Then $F$ is strategy-proof if and only if $(i)$ $F_{i}(\cdot, R_{-i}):[\underline{R}, \overline{R}]$ is monotone and $(ii)$ for all $i$ and for all $R_{-i}$ the indirect preference correspondence $V^{F_{i}(\cdot,R_{-i})}:[\underline{R}, \overline{R}]\rightarrow \mathbb{Z}$ is continuous.   	
\end{theorem}  

\begin{proof} If $F$ is strategy-proof, then the one buyer mechanism $F_{i}(\cdot, R_{-i})$ is strategy-proof for all $i$ and for all $R_{-i}$, and then analogous to the one buyer case monotonicity of the mechanisms and continuity of the indirect preference correspondences follow. Conversely, monotonicity of the mechanisms and continuity of the indirect preference correspondences imply that the one buyer mechanisms $F_{i}(\cdot, R_{-i})$ is strategy-proof for all $i$ and for all $R_{-i}$. But then this is just the definition of $n-$ buyer strategy-proof mechanisms.              
\end{proof}

\noindent No we consider the optimal strategy-proof mechanisms for $n-$buyer environment. By $F_{l}$ now we mean that $\#Rn(F_i(\cdot, R_{-i}))\leq l$ for every $i$ and every $R_{-i}$.
Also let $\mathcal{F}_{l}$ be the set of all such mechanisms. Our approach is to carry the insight from the procedure of computing the $1-$buyer optimal mechanism to the $n-$buyer environment.      

Given $\mathcal{F}_{l}$ one approach is to find a strategy-proof and individually rational mechanism $F_{i}^*(\cdot, R_{-i} )$ that maximizes expected revenue from $i$ and at $R_{-i}$. We have the tools from the $1-$buyer environment to do this.
But this approach has a problem. From the one buyer case we know that the geometry in terms of bundles and the special preferences need not entail a mechanism. This did not create any problem in the one buyer case since we can throw away some bundles and create a strategy-proof mechanism. In the $n-$buyer case it is not clear which bundles to discard. In particular, the collection $\{F_{i}^*(\cdot, R_{-i} )\mid i=1,\ldots, n, R_{-i}\in \times_{j\neq i} [\underline{R},\overline{R}]\}$ need not produce a mechanism. We may not be able to make the sum of probabilities across buyers less than or equal to $1$. As an example consider Example \ref{ex:optimal_positive_difference}. In the example for $l=3$, the optimal $1-$buyer mechanism is a deterministic mechanism, and further the optimal mechanism sells the object to all valuations in $[1,2]$. Now consider an $n-$buyer environment where buyers valuations are drawn independently and identically across with uniform distribution from $[1,2]$. Let $n=2$. Consider a profile $(1.1,1.2)$. Then, by the optimal mechanism for $F_{1}(\cdot, 1.2)$s and $F_{2}(1.1, \cdot)$s    give $Q_{1}(1.1,1.2)=1=Q_{2}(1.1,1.2)$. Thus we need restrictions on mechanisms so that the procedure to compute the optimal mechanism for $1-$buyer environment can be extended to $n-$buyer environments.       
One possible way to do this is adopt an axiomatic way. We consider the following axiom.                           
\begin{defn}\rm Let $F$ be an $n-$buyer mechanism. We say that $F$ to be 
	{\bf lower-efficient } if $R_{i}\precsim R_{j}$ for some $i\neq j$, then 
	$Q_{i}(R_{i}, R_{-i})=0$.   
\end{defn}    

\noindent Lower-efficiency requires that inefficient buyers do not get the object, but it does not mean that the mechanism is efficient. If $R_{i}\prec R_{j}$, then $T_i<T_j$ where $(0,0) I_i (T_i,1)$ and $(0,0) I_j (T_j,1)$. Intuitively, this means that buyer $j$ wishes to pay more for the good instead of not having the good, i.e., ``value'' the good more, than buyer $i$. Hence, lower-efficiency says unless a buyer has the highest ``value'' for the good, she does not get the object. Lower efficiency does not affect the revenue if probabilities of ties is zero. 
Not all mechanisms are lower-efficient. Consider an ordered set of bundles, and let a fixed buyer choose her best bundle from this set. Such a mechanism maybe called dictatorial.   
A dictatorial mechanism is strategy-proof, and not lower-efficient.   The proof of the following proposition is immediate. 

\begin{prop}\rm Let $F$ be an $n-$buyer mechanism which is lower-efficient. 
	Let for all $i$, and for all $R_i$, $R_{i}^{*}=\max\{R_{j}\mid j\neq i\}$. Let for all $i$, the one buyer mechanism $F_{i}(\cdot, R_{-i})$ restricted to $[R_{i}^*, \overline{R}]$ is strategy-proof and individually rational. Then the mechanism $F$ is strategy-proof and individually rational.        
	\label{prop:lower_efficient}
\end{prop}

\begin{proof} By individual rationality and lower-efficiency $F_{i}(R_{i}^*, R_{-i})=(0,0)$. Thus by individual rationality and lower-efficiency 
	$F_{i}(R_{i},R_{-i})=(0,0)$ for all $R_{i}\prec R_{i}^{*}$. Since
	$F_{i}(\cdot, R_{-i})$ is strategy-proof and individually rational 
	in $[R_{i}^{*}, \overline{R}]$, by the single-crossing property, $F(R_{i}, R_{-i})R_{i}(0,0)$ for all $R_{i}$ such that $R_{i}^*\prec R_{i}$.
	Also, by the single-crossing property $(0,0)R_iF_{i}(R_{i}',R_{-i})$ for all $R_{i}\prec R_{i}^{*}\prec R_{i}'$. Thus, $F$ is strategy-proof and individually rational.   
\end{proof}

\noindent An immediate advantage of lower-efficiency is that we do not need to assume any extra mathematical structure on the set of mechanisms $\mathcal{F}_{l}$. We can compute the optimal mechanism in $[R_{i}^*,\overline{R}]$ as in the one buyer scenario and extend the mechanism to the $n$-buyer environment immediately. Given $l$, we can instead look at $F_{i}(\cdot, R_{-i})$ and solve for the optimal mechanism first. The optimal mechanism for buyer $i$ will depend on the parameter $R_{i}^{*}$, i.e., information about the optimal mechanism   for one $R_i^*$ maybe enough to know the optimal mechanisms for buyer $i$ at all $R_{-i}$. Given $l$,  let $F_{i}^{*}(\cdot, R_{-i}):[\underline{R}, \overline{R}]$ be an optimal mechanism. The extension of $F_{i}^*$ from $[R_{i}^{*}, \overline{R}]$ is as described in the proof Proposition \ref{prop:lower_efficient}. Then the optimal $n-$buyer mechanism by using the 1 buyer optimal mechanisms is defined as: for all $i$ for all $(R_1,\ldots, R_{n})$,  $F_i(R_{1}, \ldots, R_{2})=F_{i}^*(R_i, R_{-i})$. Then $F$ is strategy-proof and individual, and within the set of lower-efficient mechanisms it is optimal.       
We may interpret $R_{i}^{*}$ in the line of a reserve price, and in this context we demonstrate how lower-efficiency is related to \citep{myer}.    
We assume two buyers. Assume that both buyers' valuations are drawn independently and identically from the same distribution with support $[\underline{\theta},\overline{\theta}]$.
Let the $\Gamma$ and $\gamma$ denote the distribution and density respectively. 
Let $(\theta_1, \theta_2)$ denote a profile of valuations.  
Consider
$\theta_i>\theta_j$. Thus, let by lower-efficiency $Q_{i}(\theta_i,\theta_j)>0\implies \theta_i>\theta_j$. 
Thus, consider the interval $[\theta_j,\overline{\theta}]$ to study buyer $i$.
Assume that $\Gamma$ satisfies the increasing hazard rate.  
By Corollary \ref{cor:qlin_opt} the optimal way to sell the object to buyer $i$ is    

$$T^{*}_i(\theta_i,\theta_j) = \begin{cases}
	\max \{\theta^{*}, \theta_j\}, & \text{if $\theta_i>\max \{\theta^{*}, \theta_j\}$;}\\
	$0$, & \text{if $\theta_i< \max \{\theta^{*}, \theta_j\}$.}
\end{cases}$$

$$Q_i^{*}(\theta_i,\theta_j) = \begin{cases}
	1, & \text{if $\theta_i>\max \{\theta^{*}, \theta_j\}$;}\\
	$0$, & \text{if $\theta_i< \max \{\theta^{*}, \theta_j\}$.}
\end{cases}$$

\noindent The reasoning behind above mechanism is as follows. If there were no other buyer, then buyer $i$ would have obtained the object for all valuations above $\theta^{*} \in [\underline{\theta},\overline{\theta}]$. 
Consider the optimization problem in Corollary \ref{cor:qlin_opt}    
$\theta_i[1-\Gamma(\theta_{i})]$ where $\theta_{i}\in [\underline{\theta},\overline{\theta}]$. 
The critical point is given by 
$\theta_{i}=\frac{1-\Gamma(\theta_i)}{\gamma(\theta_{i})}$. The solution to this equation is $\theta^{*}$ that is obtained in Corollary \ref{cor:qlin_opt}.
We argue that $\theta_i[1-\Gamma(\theta)]$ is increasing up to $\theta^{*}$
and decreasing after $\theta^*$. Define $\psi(\theta_i)=\theta_i-\frac{1-\Gamma(\theta_i)}{\gamma(\theta_{i})}$. 
This function is called virtual valuation in \citep{myer}. 
We have seen $\psi(\theta^*)=\theta^*-\frac{1-\Gamma(\theta^*)}{\gamma(\theta^*)}=0$.
Since the hazard rate is increasing, $\psi$ is increasing. Thus if $\theta_i<\theta^{*}$, then 
$\theta_i<\frac{1-\Gamma(\theta_i)}{\gamma(\theta_i)}$.       
Now $\frac{d \theta_{i}[1-\Gamma(\theta_{i})]}{d\theta_{i}}>0$ if $\theta_{i}<\frac{1-\Gamma(\theta_{i})}{\gamma(\theta_{i})}$. Thus the derivative of $\theta_i[1-\Gamma(\theta_i)]$ is positive if $\theta_i<\theta^{*}$.  
Since $\psi$ is increasing and $\psi(\theta^*)=0$, if $\theta^*<\theta_i$, then 
$\frac{1-\Gamma(\theta_i)}{\gamma(\theta_)}<\theta_i$.
Further, $\frac{d \theta_{i}[1-\Gamma(\theta_{i})]}{d\theta_{i}}<0$ if $\theta_{i}>\frac{1-\Gamma(\theta_{i})}{\gamma(\theta_{i})}$.
Thus the derivative of $\theta_i[1-\Gamma(\theta_i)]$ is negative if $\theta_i>\theta^{*}$.
That the optimal mechanism does not charge more than $\max\{\theta^*,\theta_j\}$ from buyer $i$ if $\theta_i>\theta_j$.  
Now note that the condition $\theta_i>\max\{\theta^*, \theta_j\}$ is same as the condition 
$\psi(\theta_i)>\psi(\theta_j)$ and $\psi(\theta_i)>0$, since $\psi$ is increasing and $\psi(\theta^*)=0$. Now if $\theta^{*}<\theta_{j}<\theta_i$, the mechanism fixes the payment to be
$\theta_j$ and not $\theta^{*}$, even though the revenue is sub-optimal compared to the one buyer scenario and the payment $\theta^*$ is feasible. If the payment is $\theta^*$ or less than $\theta_j$, then for $\theta_i'$ such that $\theta^*<T_i(\theta_i,\theta_j)<\theta_i'<\theta_j$
we obtain $\theta_i'Q_i(\theta_i',\theta_j)-T_i(\theta_i',\theta_j)=0<
\theta_i'Q_i(\theta_i,\theta_j)-T_i(\theta_i,\theta_j)=\theta_i'-T_i(\theta_i,\theta_j)$. This contradicts strategy-proofness. 
The first equality follows by lower-efficiency.  
Thus the multi-buyer mechanism  described above is exactly the mechanism optimal mechanism if we had followed the techniques in \citep{myer}.
Hence, we also provide a geometric understanding of the optimal mechanisms in \citep{myer}. 

One possible suggestion could be to consider Bayesian incentive compatibility so that dependence on $R_{-i}$ to find $F_{i}^*$ can be avoided. We consider two classes of utility 
functions $\theta \times \text {probability of win}-\text{payment}$ and $\theta \times \text {probability of win}-(\text{payment})^{2}$. Let the valuations be distributed independently and identically across buyers with joint density $\gamma(\theta_1,\ldots, \theta_n)=\times_{i=1}^{n}\gamma(\theta_i)$, and let $\gamma(\theta_{-i})=\times_{j\neq i}^{n}\gamma(\theta_j)$ 
 The Bayesian incentive compatibility for these two classes of utility functions are defined as
  
for all $\theta_i$, $\int_{\theta_{-i}}[\theta_i Q_{i}(\theta_i,\theta_{-i})-T_{i}(\theta_i,\theta_{-i})]\gamma(\theta_{-i})d(\theta_{-i})\geq \int_{\theta _{-i}}[\theta_{-i} Q_{i}(z_i,\theta_{-i})-T_{i}(z_i,\theta_{-i})]\gamma(\theta_{-i})d(\theta_{-i})$, 

and 
for all $\theta_i$, $\int_{\theta_{-i}}[\theta_i Q_{i}(\theta_i,\theta_{-i})-T_{i}(\theta_i,\theta_{-i})^{2}]\gamma(\theta_{-i})d(\theta_{-i})\geq \int_{\theta _{-i}}[\theta_{-i} Q_{i}(z_i,\theta_{-i})-T_{i}(z_i,\theta_{-i})^{2}]\gamma(\theta_{-i})d(\theta_{-i})$.

\noindent Let $t_{i}(\theta_i)=\int_{\theta_{-i}}T_{i}(\theta_i,\theta_{-i})\gamma(\theta_{-i})d(\theta_{-i})$, and $q_{i}(\theta_i)=\int_{\theta_{-i}}Q_{i}(\theta_i,\theta_{-i})\gamma(\theta_{-i})d(\theta_{-i})$. Then incentive compatibility reduces to $\theta_i q_{i}(\theta_i)-t_{i}(\theta_i)$. Note that
in the case of linear preference, whether we consider utility of the buyer before or after taking average the evaluation of a buyer's, i.e., her probability of win and her payment, by using the same linear preference. Things change for utility functions with positive income effects: let   $t_{i}(\theta_i)=\int_{\theta_{-i}}T_{i}(\theta_i,\theta_{-i})^{2}\gamma(\theta_{-i})d(\theta_{-i}))$. The incentive compatibility becomes  
$\theta_iq_{i}(\theta_i)-t_{i}(\theta_i)\geq \theta_iq_{i}(z_i)-t_{i}(z_i)$.
But this utility function does not incorporate positive income effect.  
Along with creating a confusion about whether to consider linear utility function or utility function with positive income effect,       
analyzing optimal mechanisms becomes harder.  Latter happens because   
expected payment received by the seller from  $\theta_i$ is $\widehat{t_i}(\theta_i)=\int_{\theta_{-i}}T_{i}(\theta_i,\theta_{-i})\gamma(\theta_{-i})d(\theta_{-i}))\neq t_{i}(\theta_i)$ in general. However, if we assume the mechanism to be lower-efficient, we can circumvent this confusion. It is because by lower-efficiency, and individual rationality in $[\theta_{i}^*, \overline{\theta}]$, all other buyers' allocations are $(0,0)$, and then we can study strategy-proof and individually rational mechanism by applying the primitive of buyer $i$.   
We end this section with a remark about optimal mechanism with quality.   

\begin{remark}\rm Suppose on the vertical axis of $\mathbb{Z}$ instead of probability of win we have quality. As we discussed earlier in this case we may not be able to write down a preference by using a known formula.  However, lower-efficiency is well defined. There is yet another approach we may consider. We consider an axiom related to priority over the set of buyers. Let $q$ denote shares.
An example of shares can be thought of as selling rights by a firm to a vendor in such a way that the vendor can provide specific services to the
customers.{\footnote{\url{https://www.lawinsider.com/dictionary/selling-rights\#:~:text=Selling20Rights20means20the20right,Customer20under20this20Framework20Agreement}}}.  
Consider a firm that produces soft drinks. The management of the firm may grant selling rights of its products to a party in such a way that only certain  specified brands of soft drinks can be sold by the licensee. In an abstract sense, if we think the set of all brands of soft drinks produced by the firm to be 1, then selling rights to sell certain brands can be considered as selling  a share.	
Let there be two buyers and the seller wants to sell shares first to buyer $1$ and if anything is left then it will be sold to buyer $2$. This procedure is not lower-efficient. Consider the one buyer problem of finding the optimal mechanism for buyer $1$ whose preferences lie in $[\underline{R},\overline{R}]$. Let buyer $1$'s preference be $R_1$ and suppose $F_{1}(R_1,R_2)=(T_1,Q_1)$ for all $R_2$. Then given $R_1$ consider the one buyer problem of finding the optimal mechanism for buyer $2$ where $Q_2\in [0,  1-Q_1]$. This defines a strategy-proof mechanism. By borrowing a terminology from social choice theory we may call such a procedure a serial dictatorial. Such priority over the set of buyers may arise from the differences in the reputations of the buyers to fulfill the contractual obligations such as making timely payments, see \citep{Houser} and \citep{Englemann} for theory and importance of reputation in auctions. 
Instead of priority, fairness could be an external restriction on the mechanisms while selling shares. For example, let each buyer can obtain at most $\frac{1}{n}$, where $n$ is the number of buyers, fraction of the total share. In this case $\{F_{i}^{*}(\cdot, R_{-i})\mid i=1,\ldots, n, R_{-i}\in\times_{j\neq i}[\underline{R}, \overline{R}]\}$ define a mechanism. To see this, fix $R_{-i}$ and let $F_{i}^{*}(\cdot, R_{-i})$ be an optimal $1-$buyer mechanism at $R_{-i}$. 
Define $F(R_1,R_2,\ldots R_{n})=(F_{1}^{*}(R_1, R_{-1}),F_{2}^{*}(R_2, R_{-2}), \ldots, F_{n}^*(R_{n}, R_{-n}))$. Since $Q_{i}^{*}(R_i, R_{-i})\leq \frac{1}{n}$, $\sum_{i=1}^{n}Q_{i}^*(R_{i}, R_{-i})\leq 1$.

\end{remark}

\section{\bf Concluding Remarks}   
\label{sec:con}

This paper analyzes strategy-proof and individually rational mechanisms defined on CM single-crossing domains. 
We provide a characterization of strategy-proof mechanisms that are monotone, and the indirect preference correspondences that correspond to the mechanisms are continuous.  
Our characterization brings out the geometry of strategy-proof mechanisms ensued from the single-crossing property of the CM preferences. Instead of searching in the space of mechanisms for an optimal mechanism, it is enough to search in the space of the geometry. The geometry that we are referring to is just a finite product metric spaces of bundles and the CM single-crossing domain under consideration. We believe we have provided a computable optimization program that maximizes expected revenue in space of geometry of strategy-proof and individually rational mechanisms. Our belief stems from the fact that computations on abstract topological spaces is an active area of research currently. Instead of looking at the $1-$buyer environment as a special case ob an $n-$buyer environment, we provide an axiomatic extension of the $1-$buyer environment to an $n-$buyer environment. In situations where the number of buyers is large, and the number keeps changing quite frequently our axiomatic extension could be useful.           
A standard approach in mechanism design theory involves finding a payment rule that implements an allocation rule that makes revealing true preference a dominant strategy. We consider allocations and payments together as bundles or alternatives and study strategy-proof mechanisms. In this sense our approach can be considered social choice theoretic.

\begin{appendix}
	\section{{\bf Appendix 1}}

\noindent Let $\square(z)=\{x\mid x\leq z\}$, and recall Definition \ref{defn:cut}.  
In words, the indifference curve of $R''$ through $z$ lies above the indifference curve for $R'$ through $z$ in $\square(z)$ if both indifference curves are viewed from the horizontal axis $\Re_{+}$. By applying arguments similar to ones in  
\citep{Goswami} it can be  established that if $R''$ cuts $R'$ from above at some $z\in \mathbb{Z}$, then $R''$ cuts $R'$ from above at every $z\in \mathbb{Z}$. In \citep{Goswami} preferences are assumed to have strictly convex upper contour sets. However Proposition $3$ in \citep{Goswami} which proves if $R''$ cuts $R'$ from above at some bundle, then $R''$ cuts $R'$ from above at every bundle does not use convexity of preferences.
Thus we define the following order on $\mathcal{R}^{cms}$: 
for all $R',R''\in \mathcal{R}^{cms}$ we say $R'\prec R''$ if $R''$ cuts $R'$ from above. We consider the order topology on $\mathcal{R}^{cms}$ by applying this order.   
We recall a few definitions. An order $\prec$ in any set $L$ is called {\bf simple or linear} if  
it has the following properties (see \citep{Munkres}):
$(1)$ (Comparability) For every $\alpha$ and $\beta$ in $L$ for which $\alpha\neq \beta$, either $\alpha \prec \beta$ or $\beta\prec \alpha$,
$(2)$ (Non-reflexivity) For no $\alpha$ in $L$ does the relation $\alpha \prec \alpha $ hold,
(3) (Transitivity) If  $\alpha \prec \beta $ and $\beta \prec \gamma$, then $\alpha\prec \gamma$. An important class of linearly ordered sets are linear continuums. We define this notion next.    

\begin{defn}[Linear continuum]\rm (see \citep{Munkres}) A simply ordered set $L$ having more than one element is called a linear continuum if the following hold:
	$(1)$ $L$ has the least upper bound property, i.e., every bounded subset, i.e.,  bounded according to the order $\prec$ on $L$, has the least upper bound in $L$, (see \citep{Rudin}), $(2)$ if $\alpha \prec \beta$, then there exists $\gamma$ such that $\alpha\prec \gamma \prec \beta$, $\alpha, \beta,\gamma $ are in $L$.   
\end{defn}

\noindent In Theorem \ref{thm:lin_cont} we show that $\mathcal{R}^{cms}$ is a linear continuum. 
To define the order topology for any two preferences $R', R''\in \mathcal{R}^{cms}$ with $R'\prec R''$ define open interval $]R', R''[=\{R|R'\prec R \prec R''\}$. The collection of open intervals form a basis on $\mathcal{R}^{cms}$. The topology generated by this basis is the order topology on $\mathcal{R}^{cms}$.
In this topology $\mathcal{R}^{cms}$ is connected and closed interval $[R',R'']=\{R\mid R'\precsim R \precsim R''\}$ is compact, here $R'\precsim R$ means either $R'\prec R$ or $R'=R$, $]-\infty, R[$ and $]-\infty, R]$ denote open and closed intervals respectively that are not bounded below, $]R,\infty[$ and $[R,\infty]$ denote open and closed intervals respectively that not bounded above,
see \citep{Munkres} for details. In this order topology $\mathcal{R}^{cms}$ is homeomorphic to $\Re$ and thus metrizable. Further, the homeomorphism is an order preserving bijection with the real line. The next lemma formalizes these observations. Before proceeding we note that intervals in all ordered sets  in this paper are written by using the braces $]$ and $[$. For example, if $\alpha,\beta \in \Re $ and $\alpha<\beta$, then $]\alpha,\beta]$ denotes the left open and right closed interval.

\begin{theorem}\rm  There is an order preserving homeomorphism, denoted by $h$, between  $\mathcal{R}^{cms}$ with the order topology and $\Re$ with the standard Euclidean metric topology.  Due to this order preserving bijection $\mathcal{R}^{cms}$ is a linear continuum.     
	
	\label{thm:lin_cont}
\end{theorem}

\begin{proof} We construct an order preserving bijection between $\mathcal{R}^{cms}$ with order $\prec$ and an open interval in the real line $\Re$. Consider $z=(0,0)$ and let $C_{\delta}(0,0)$ be the circle of radius $\delta<1$ with origin $(0,0)$. Consider $\mathbb{A}=[C_{\delta}(0,0)\cap \mathbb{Z}]\setminus \{(\delta,0)\cup (0,\delta)\}$. 
	That is, $\mathbb{A}$ is the quarter circle with radius delta and center $(0,0)$ intersected with $\mathbb{Z}$ excluding the end points. 
	Let $h':\mathcal{R}^{cms}\rightarrow \mathbb{A}$ by $h'(R)=(t,q)$ where $(t, q)I(0,0)$. 
	By richness of $\mathcal{R}^{cms}$, for every $z\in \mathbb{A}$ there exists $R\in \mathcal{R}^{cms}$ such that $zI(0,0)$. Thus, $h'$ is onto. By the single-crossing property $h'$ is one-one.

	Further, $\mathbb{A}$ is homeomorphic to the open interval $]0,\delta[$ in $\Re$, where the homeomorphism is just the projection $\pi:\mathbb{A}\rightarrow ]0,\delta[$ defined by $\pi(t,z)=t$. Then we obtain an order preserving homeomorphism between 
	$\mathcal{R}^{cms}$ and $]0,\delta[$ defined by the composition $\pi~ o~ h'$, call the homeomorphism $h$. Since $]0,\delta[$ has the least upper bound property so does $\mathcal{R}^{cms}$. Also $\mathcal{R}^{cms}$
	has the second property of a linear continuum by richness.  
	For any $R',R''\in \mathcal{R}^{cms}$, define $d(R',R'')= |h(R')-h(R'')|$ where $|\cdot|$ is the standard Euclidean metric on $\Re$. 
\end{proof}

\noindent An implication of Theorem \ref{thm:lin_cont} is that convergence of a sequence in $\mathcal{R}^{cms}$ is equivalent to saying that monotone sequences converge. A sequence of preferences is denoted by 
$\{R^{n}\}_{n=1}^{\infty}$. A monotone decreasing sequence means any $n$, $R^{n+1}\precsim R^{n}$, and its convergence to $R$ is denoted by $R^{n}\downarrow R$. A monotone 
increasing sequence means for any $n$, $R^{n}\precsim R^{n+1}$, and its convergence to $R$ is denoted by $R^{n}\uparrow R$. In general convergence to $R$ is denoted by $R^{n}\rightarrow R$.
Further, $R^{n}\rightarrow R \iff h(R^{n})\rightarrow h(R)$. It is known that a sequence in $\Re$ converges if and only if every monotone subsequence of the sequence converges to the same limit. Since $h$ is an order preserving homeomorphism the same is true for $\mathcal{R}^{cms}$.   
Another very important, perhaps the most critical, implication of Theorem \ref{thm:lin_cont} is that 
$\mathcal{R}^{cms}$ in the order topology is an one dimensional topological manifold. Later we shall give examples to show that even if $\mathcal{R}^{cms}$ is topologically one dimensional, from the perspective  of utility representation of these preferences in terms parametric classes $\mathcal{R}^{cms}$ is multidimensional. 
We end this section with two technical lemmas that we use later. By $cl(B)$, we mean closure of the set $B$. The next lemma shows that the infimum and the supremum of a set in $\mathcal{R}^{cms}$ are in the closure of the set. Infimum and supremum are defined by using the order $\prec$.

\begin{lemma}\rm If $R_{0}=\inf_{\prec} U$ where $ U\subseteq \mathcal{R}^{cms}$, then $R_{0}\in cl(U)$. 
	Analogously, if $R^{0}=\sup_{\prec} U$ where $U \subseteq \mathcal{R}^{cms}$, then $R^{0}\in cl(U)$.  	  
	\label{lemma:closure}
\end{lemma}	

\begin{proof} Let $R_{0}=\inf_{\prec} U$. Suppose there exists $R\neq R_{0}$ such that $]R_{0},R[\cap U =\emptyset$. Since $R_{0}$ is the infimum, $R_{0}\in U$. Thus $R_{0}\in cl(U)$. If for every $R\neq R_{0}$, $]R_{0},R[\cap U \neq \emptyset$, then $R_{0}$ is a limit point of $U$ and thus lies in $cl(U)$.  
The argument for the supremum is analogous.   	   
\end{proof}

\noindent The Lemma \ref{lemma:preference_conv} proves a closure property of the order topology on $\mathcal{R}^{cms}$: indifference in the limit of any sequence of preferences is preserved. Lemma \ref{lemma:preference_conv} is very important when proving the constraint set of the optimization program to find the optimal mechanism in Theorem \ref{thm:optimal} is compact.

\begin{lemma}\rm Let $\{R^{n}\}_{n=1}^{\infty}\subseteq \mathcal{R}^{cms}$ be a sequence such that $R^{n}\rightarrow R$. Let $(t^{1n},q^{1n})\rightarrow (t^1,q^1),(t^{2n},q^{2n})\rightarrow (t^2,q^2)$, where $\{(t^{in},q^{in})\}_{n=1}^{\infty}\subseteq \mathbb{Z}$, $i=1,2$.
	If $(t^{1n},q^{1n})I^{n}(t^{2n},q^{2n})$, then $(t^{1},q^{1})I(t^{2},q^{2})$. 	   
	\label{lemma:preference_conv}
\end{lemma}

\begin{proof}The following claim is required to prove the lemma. Consider a compact subspace $K=[0, \bar{t}]\times [0,1]$
	whee $0<\bar{t}<\infty$.

	\begin{claim}\rm Each preference in $\mathcal{R}^{cms}$ can be represented by a continuous function $f$ such that if $R^{n}\rightarrow R$, then $f^{n}\rightarrow f$ uniformly in $K$.  
		\label{claim:uniform}	   
	\end{claim}
	
	\noindent{\bf Proof of the Claim \ref{claim:uniform}:} 
	Any preference $R\in \mathcal{R}^{cms}$ is completely defined by the equivalence
	classes formed by its indifference sets. Since $R\in \mathcal{R}^{cms}$, by the monotonicity each equivalent class can be identified uniquely with a bundle $(t, 1)$, where $t\in [0,\infty[$.
	An equivalence class of $R$ denoted by $[t]_{R}$ is: 
	$[t]_{R}=\{(t',q')\mid (t',q') I (t,1), (t',q')\in \mathbb{Z}\}$.
	Let $f_{R}(t',q')=t$ for all $(t', q')\in  [t]_R$. 
	That is for $(t',q'),(t'',q'')\in \mathbb{Z}$ we have $(t',q')R(t'',q'')$ if and only if $f_{R}(t',q')\leq f_{R}(t'',q'')$.{\footnote{We do not require $f_{R}$ to be a utility representation.}}

	This function is continuous. 
	To see this let $(t^{'n},q^{'n})\rightarrow (t',q')$. 
	We want to show that $\lim_{n\rightarrow \infty}f_{R}(t^{'n},q^{'n})=f_{R}(t',q')$. 
	Let $(t^{'n},q^{'n})I(t^{n},1)$ and thus $f(t^{'n},q^{'n})=t^{n}$. 
	Further, let $(t',q')I(t,1)$ and thus $f_{R}(t',q')=t$.  
	Let by way of contradiction $t^{n}$ does not converge.
	Consider the situation where $\{t^{n}\}_{n=1}^{\infty}$ is an unbounded sequence. 
	Then for every $M>0$ there is $t^{n}>M$. 
	Since $(t',q')I (t,1)$, by monotonicity of $R$  
	there is $(t^{n},1)$, $n$ large enough so that $(t',q')P(t^{n},1)$.  
	Since $(t^{'n},q^{'n})\rightarrow (t',q')$ and $R$ is continuous there is $N$ such that 
	for all $m\geq N$, $(t^{m},q^{m})P(t^{n},1)$. 
	Since $\{t^{n}\}_{n=1}^{\infty}$ is  unbounded there is $m>N$ such that $t^{n}<t^{m}$. 
	By monotonicity of 
	$R$, $(t^n,1)P(t^m,1)$ and thus by transitivity of $R$ we obtain $(t^{m},q^{m})P(t^{m},1)$, this contradicts the definition of $f_{R}$.

	Thus  let $\{t^{n}\}_{n=1}^{\infty}$ be bounded and does not converge. 
	Let $\{t^{n_k}\}_{k=1}^{\infty}$ be a convergent subsequence of $\{t^{n}\}_{n=1}^{\infty}$. We show that the subsequence converges to $t$. Let by way of contradiction 
	$t^{n_k}\rightarrow t^{*}\neq t$. Without loss of generality let $t^{*}<t$. Consider
	$[t^{*}-\epsilon, t^{*}+\epsilon]~\cap~ [t-\epsilon, t+\epsilon]=\emptyset$. 
	For $K$ large enough for all $k\geq K$, $t^{n_k}\in ]t^{*}-\epsilon, t^{*}+\epsilon[$.
	Now $(t^{'n_k},q^{'n_k})I (t^{n_{k}},1)$. By monotonicity of $R$, $(t^{'n_k},q^{'n_k})\in [0, t^{*}+\epsilon]\times [0,1]$. Since $\{(t^{'n},q^{'n})\}_{n=1}^{\infty}$, and thus $\{(t^{'n_k},q^{'n_k})\}_{k=1}^{\infty}$ converges to $(t',q')$ we have  $(t',q')\in [0, t^{*}+\epsilon]\times [0,1]$.  
	By monotonicity of $R$, $(t-\epsilon,1)P (t,1)I(t',q')$. By continuity of $R$ there is $B(\epsilon, (t',q'))$ such that $(t-\epsilon,1)Pz$ for all $z\in B(\epsilon, (t',q'))$.  Here $B(\epsilon,(t',q')) = \{(t,q)\mid \sqrt{(t-t')^{2}+(q-q')^{2}} <\epsilon\}$. Since $(t^{'n},q^{'n}) \rightarrow (t',q')$,  for some $N$ for all $n\geq N $ we have $(t-\epsilon,1)P(t^{'n},q^{'n})$. By monotonicity of $R$, $(t^{'n_k},q^{'n_k})$s are not indifferent to $(t^{n_k},1)$ where $t^{n_k}<t^{*}+\epsilon<t-\epsilon$. This is a contradiction.    
	This is a contradiction. 
	Thus $f_{R}$ is continuous.

	Define $f_{R}$ similarly for all $R$. Then by richness, the single-crossing property and
	the definition of $f_{R}$, $f_{R^n} \rightarrow  f_R$ pointwise.
	To see this let $R^{n}\rightarrow R$ and consider $(t,q)\in \mathbb{Z}$. Without loss of generality
	let $\{R^{n}\}_{n=1}^{\infty}$ be an increasing sequence. 
	Let $f_{R}(t,q)=t'$, i.e., $(t,q)I(t',1)$. Consider $]t'-\epsilon, t']$.
	By richness there is $R^N$ such that $(t,q)I^{N}(t^{N},1)$ for some $t^{N}\in ]t'-\epsilon,t']$.
	By the single-crossing property for all $n\geq N$, $(t,q)I^{n}(t^{n},1)$, where $t^{n}\in ]t'-\epsilon,t']$. 
	Thus for all $n\geq N$, $f_{R^{n}}(t,q)\in ]t'-\epsilon,t']$. This establishes the required pointwise convergence.      
	Also note that if $R'\prec R''$, then $f_{R'}(t,q)<f_{R''}(t,q)$. 
	If $f_{R^n} \rightarrow  f_R$ pointwise,  then without loss of generality we can consider monotone sequences. Then
	by Theorem $7.13$ in \citet{Rudin} $f_{R^n} \rightarrow  f_R$ 
	converges to $f_{R}$ uniformly. This establishes Claim \ref{claim:uniform}. 
	To see that it is enough to consider monotone subsequence of $\{f_{R^{n}}\}_{n=1}^{\infty}$.
	Let $M_{n}=\sup_{x\in K}|f_{R^{n}}(x)-f_{R}(x)|$. For uniform convergence we need to show that 
	$\lim_{n\rightarrow \infty}M_{n}=0$. By the single-crossing property every monotone subsequence of $\{M_{n}\}_{n=1}^{\infty}$ corresponds to a monotone subsequence of $ \{ f_{R^{n}}  \}_{n=1}^{\infty}$ and vice versa. For example, $R'\prec R'' \prec R\iff ~\text{for all}~x\in K, f_{R'}(x)< f_{R''}(x) <f_{R}(x)\iff ~\text{for all}~x\in K, f_{R}(x)-f_{R''}(x)<f_{R}(x)-f_{R'}(x)$. 
	The first equivalence follows because if $R''$ cuts $R'$ from above at some $z\in \mathbb{Z}$, then $R''$ cuts $R'$ from above at every $z\in \mathbb{Z}$.
	Since $\{M_n\}_{n=1}^{\infty}$ converges to $0$ if 
	and only if all its monotone subsequences converge to $0$, the required uniform convergence follows from Theorem $7.13$ in \citet{Rudin}.                
	We require another claim.
	
	\begin{claim}\rm 
		Let $f_{R^n}
		\rightarrow f_R$ uniformly on a compact set $K
		\subseteq \mathbb{Z}$ 
		and $f_{R^n}$ be continuous. Let $\{x^{n}\}_{n=1}^{\infty}\subseteq K$. If $x^n \rightarrow x$,
		then $\lim_{n\rightarrow \infty} f_{R^n}(x^{n})=f_{R}(x) $
		
		\label{claim:uni_converge}
	\end{claim}

	\noindent{\bf Proof of Claim \ref{claim:uni_converge}:} Since $K$ is compact, $x\in K$.  
	Fix $\epsilon>0$. By uniform convergence there is $N_1$ such that
	$\sup_{z\in K}|f_{R^n}(z)-f_R(z)|$
	$<\frac{\epsilon}{2}$ for all $n\geq N_1$.
	By continuity of $f_R$, since uniform convergence preserves continuity, there is $N_2$ such that
	$|f_R(x^n)-f_R(x) |< \frac{\epsilon}{2}$
	for all $n\geq N_2$.  
	Let $N=\max\{N_1,N_2\}$. Now
	$|f_{R^n}(x^n)-f_R(x)|=
	|f_{R^n}(x^n)-f_R(x^n)+ f_R(x^n)- f_R(x)|\leq|f_{R^n}(x^n)-f_R(x^n)|+ |f_R(x^n)- f_R(x)|< \epsilon$ for all $n\geq N$.
	Hence the proof of the claim follows.
	
	Back to the proof of lemma. Since the sequence of bundles converge we can assume that they lie in a compact set $K$.  By Claim \ref{claim:uniform}  and Claim \ref{claim:uni_converge} $f_{R^n}(t^{1n},q^{1n})\rightarrow f_{R}(t^{1},q^{1})$
	and $f_{R^n}(t^{2n},q^{2n})\rightarrow f_{R}(t^{2},q^{2})$. 
	Since $f_{R^n}(t^{1n},q^{1n})=f_{R^n}(t^{2n},q^{2n})$, $f_{R}(t^{1},q^{1})=f_{R}(t^{2},q^{2})$.
	Therefore, $(t^{1},q^{1})I(t^{2},q^{2})$. This establishes the lemma.

\end{proof}

\begin{remark}\rm Being a linear continuum closed intervals $[\underline{R},\overline{R}]$ in $\mathcal{R}^{cms}$ in the order topology are compact. Further, being a linear continuum $\mathcal{R}^{cms}$ has the least upper bound property so that infimum and supremum are well defined. 
\end{remark}

\end{appendix}

\begin{appendix}
\section*{{\bf Appendix 2}} 

\noindent{\bf Proof of Lemma \ref{lemma:mon}:} 	Let $F:\mathcal{R}^{cms}\rightarrow \mathbb{Z}$ be a strategy-proof SCF. Let $R', R''\in \mathcal{R}^{cms}$ be two preferences such that $R'\prec R''$, and $F(R')\neq F(R'')$.
Since $F$ is strategy-proof, $F(R'')R''F(R')$. Thus, 
\begin{equation}\label{nec1}
F(R')\in LC(R'', F(R'')).
\end{equation}
Again by strategy-proofness, $F(R')R' F(R'')$. Thus,
\begin{equation}\label{nec2}
F(R')\in UC(R', F(R'')).
\end{equation}
Combining the two equations \eqref{nec1} and \eqref{nec2}, we get,
\begin{equation}\label{nec3}
F(R')\in LC(R'', F(R''))\cap UC(R', F(R'')).
\end{equation}
Since $R'\prec R''$, by monotonicity of CM preferences 
\begin{equation}\label{nec4}
\begin{split}
[LC(R'', F(R''))\cap UC(R', F(R''))\subseteq \{z| z\leq F(R'')\}].
\end{split}
\end{equation}
From equations \eqref{nec3} and \eqref{nec4} we obtain $F(R')\leq  F(R'')$.
Hence the proof of the lemma follows.

\medskip

\noindent{\bf Proof of Lemma \ref{lemma:preference_preserve}:} We note that $(i)$ is just a rewriting of
the single-crossing condition by using the order on $\mathcal{R}^{cms}$. For $(ii)$, note that
by richness there is $\widetilde{R}\in \mathcal{R}^{cms}$ such that $z'\widetilde{I}z''$. 
If $z'Pz''$, then $R\prec \widetilde{R}$; and if $z''Pz'$, then $\widetilde{R}\prec R$. Now by $(i)$ the proof of the lemma follows.   

\medskip

\noindent{\bf Proof of Lemma \ref{lemma:cont_correspondence}:}Let $F:\mathcal{R}^{cms}\rightarrow \mathbb{Z}$ be strategy-proof. Without loss of generality consider a decreasing sequence $\{R^{n}\}_{n=1}^{\infty}$
that converges to $R$.  	
From Lemma \ref{lemma:mon}, it follows that the sequence $\{F(R^n)\}_{n=1}^{\infty}$ is a decreasing sequence bounded below by $F(R)$. Thus, $\{F(R^n)\}_{n=1}^{\infty}$ converges. Let $z=(t,q)=\lim_{n\rightarrow \infty}F(R^n)$. 
Let by way of contradiction $zIF(R)$ does not hold. Then either $(i)$ $zPF(R)$ or $(ii)$ $F(R)Pz$. 

\noindent {\bf Consider Case $(i)$}: By continuity of $R$ there is an open ball $B(\epsilon,z)=\{z'\in \mathbb{Z}\mid ||z-z'||<\epsilon\}$ such that for all $z'\in B(\epsilon,z)$, $z'PF(R)$. If $z=(t,q), z'=(t',q')$ then $||z'-z||=\sqrt{(t-t')^{2}+(q-q')^{2}}$. Since $\lim_{n\rightarrow \infty }F(R^{n})=z$, there is some positive integer $N$ such that 
$F(R^{N})\in B(\epsilon,z)$ such that $F(R^{N})PF(R)$. This is a contradiction to 
strategy-proofness.

\medskip     

\noindent {\bf Consider Case $(ii)$}: Let $F(R)=(t(R),q(R))$. Since $\{R^{n}\}_{n=1}^{\infty}$ is a decreasing sequence, and $F(R)\neq z$, 
$(t(R),q(R))\leq(t,q)$. By monotonicity of $R$, $t(R)<t$. Note that if $t(R)=t$, then by contradiction hypothesis $q(R)<q$. Then by monotonicity of $R$, $(t,q)P F(R)$, and thus we are in case $(i)$. 
Back to case $(ii)$. By continuity of $R$ there is $B(\epsilon,z)$ such that 
for all $z'\in B(\epsilon,z)$, $F(R)Pz'$. Consider $[t,t']\times [q,q']\subseteq B(\epsilon,z)$.
Note $z=(t,q)\in [t,t']\times [q,q']$. Now $t(R)<t$ and $q(R)<q'$. By richness consider $R^{*}$ such that 
$F(R)I^{*}(t,q')$. By Lemma \ref{lemma:preference_preserve} $R\prec R^{*}$. 
By monotonicity of $R^{*}$, if $z'\neq (t,q'), z'\in [t,t']\times [q,q']$, then $(t,q')P^{*}z'$. 
By By Lemma \ref{lemma:preference_preserve} for all $R'$ such that $R'\prec R^{*}$,
if $z'\neq (t,q'), z'\in [t,t']\times [q,q']$, then $(t,q')P'z'$.
By the single-crossing property $F(R)P'(t,q')$. Since $\lim_{n\rightarrow \infty}F(R^{n})=z$, there is 
$R^{n}\prec R^{*}$ such that $F(R^{n})\in ]t,t'[\times ]q,q'[$. Thus $F(R)P^{n}(t,q')P^{n}F(R^{n})$. 
Hence, $F(R)P^{n}F(R^{n})$ and it is a contradiction to strategy-proofness.     

\begin{remark}\rm This proof goes through for $\mathcal{R}^{rrca}$ if we consider the domain of $F$ to be $[\underline{R},\overline{R}]\subseteq \mathcal{R}^{cms}$ since for our arguments we do not need $R^{*}\in [\underline{R},\overline{R}]$. Since the proof of the monotonicity of $F$ uses only two preferences, this proof holds if the domain of $F$ is $[\underline{R},\overline{R}]$.      
\end{remark}	

\noindent This completes the proof of the lemma. 

\medskip

\noindent{\bf Proof of Theorem \ref{thm:implies_strtagey_proof_finite_range}:} 	Without loss of generality we assume  $\#Rn(F)=6$, 
i.e., we let $Rn(F)=\{e, c,a,b,d,f\}$ and $e<c<a<b<d<f$. Assuming $\#Rn(F)=6$ is without loss of generality because the argument for the bundles that are non-extreme such as $a,b$ is independent of the number of elements in $Rn(F)$. Also the argument is same for the extreme bundles such as $e$ or $f$. 
We first prove that $F$ is locally strategy-proof in range, it is enough to prove this for the bundles $a$ and $b$.

\begin{claim} \rm ({\bf Local Strategy-proofness in range}) Let $R'\prec  R''$, $F(R')=a$ and $F(R'')=b$. Then $aR'b$ and $bR''a$.
	\label{claim:ab}
\end{claim}

\noindent{\bf Proof of Claim \ref{claim:ab}:}
Let by way of contradiction the claim is false. 
Let without loss of generality $aP''b$. 
Define $S=\{R \in \mathcal{R}^{cms}\mid F(R)=a\}$ and $T=\{R \in \mathcal{R}^{cms}\mid F(R)=b\}$. 
Let $R^0$ and $R_0$ be the supremum and the infimum of the sets $S$ and $T$ respectively under the ordering $\prec$. That is, 
$R^0=\sup_{\prec}\{R\mid F(R)=a\}$ and
$R_0=\inf_{\prec}\{R\mid F(R)=b\}$. 
We note that $R'\in S$ and $R''\in T$. Further $R'\precsim R^{0}$ and 
$R_{0}\precsim R''$.
We prove the claim in several steps.

\noindent{\bf Step $1$:}  The following statements hold: $(a)$ $R^0\precsim R''$ , 
$(b)$ $R'\precsim R_0$.

\noindent{\bf Proof of Step $1$: } 
To establish statement (a), let by way of contradiction $R''\prec R^0$. 
By monotonicity of $F$, $F(R'')=b\leq F(R)$ for all $R$ such that $R''\prec R$.       
Thus, $R''$ is an upper bound of $S$. Therefore, if $R''\prec R^{0}$, 
then $R^{0}$ is not the supremum of $S$. This is a contradiction.  
To establish $(b)$ let by way of contradiction $R_{0}\prec R'$. By monotonicity of $F$, 
$F(R)\leq a=F(R')$ for all $R$ such that $R\prec R'$.  Thus $R'$ is a lower bound of $T$.
Therefore, if $R_{0}\prec R'$, then $R_{0}$ is not the infimum of $T$.

\medskip

\noindent{\bf Proof of Step $2$:} $R_{0}=R^{0}$. 

\noindent{\bf Proof of Step $2$:} 
Let by way of contradiction  $R_{0}\neq R^{0}$. 
Thus either Case $(i)$ $R^{0}\prec R_{0}$ or Case $(ii)$ $R_{0}\prec R^{0}$. 

\medskip

\noindent Case $(i):$ Suppose $R^{0}\prec R_{0}$. Since $\mathcal{R}^{cms}$ is rich there exists a preference $\widehat{R}$ such that $R^0\prec\widehat{R}\prec R_0$. 
Since $R^{0}$ is the supremum of $S$, $F(R^{0})\nless a$, and since $R_{0}$ is the infimum of $T$,
$F(R_{0})\ngtr b$.  
Thus by monotonicity of $F$, $F(R^{0})\in \{a,b\}$ and  
$F(R_{0})\in \{a,b\}$. Again by monotonicity of $F$, $F(\widehat{R})\in \{a,b\}$. 
Let $F(\widehat{R})=a$, and then $\widehat{R}\in S$.
Since $ R^0\prec \widehat{R}$, $R^0$ is not the supremum of $S$. 
This is a contradiction.  
Alternatively, let $F(\widehat{R})=b$. 
Since $ \widehat{R}\prec R_0$, $R_0$ is not the infimum of $T$.
This is a contradiction.  
This establishes that Case $(i)$ cannot occur.

\medskip

\noindent Case $(ii):$ Suppose $ R_0\prec R^0 $. Since $\mathcal{R}^{cms}$ is rich there exists a preference $\widehat{R}$ such that $R_0\prec\widehat{R}\prec R^0$. We show $F(\widehat{R})\in \{a,b\}$.
If $F(\widehat{R})<a$, then by monotonicity of $F$, $F(R)< a$ for all $R\prec \widehat{R}$. Since $F(R')=a$, $R'$ is a lower bound on $T$. Further, $F(R_{0})<a=F(R')$, thus by monotonicity of $F$, $R_{0}\prec R'$. Thus $R_{0}$ is not the infimum of $T$, which is a contradiction.    
If $F(\widehat{R})>b$, 
by monotonicity of $F$, $F(R^{0})>b$. Since $F(R'')=b$, $R''$ is an upper bound 
on $S$. Further $b=F(R'')<F(R^{0})$, thus by monotonicity of $F$, $R''\prec R^{0}$.   
Thus $R^{0}$ is not the supremum of $S$, which is a contradiction.  
Therefore, $F(\widehat{R})\in \{a,b\}$.

If $F(\widehat{R})=a$, then by monotonicity of $F$ if $F(R)=b$, then $\widehat{R}\prec R$. 
Further, $F(R)\leq a$ for all $R$ such that $R\prec \widehat{R}$. Thus $\widehat{R}$ is a lower bound of $T$, and therefore   
$R_{0}$ is not the infimum of $T$. 
This is a contradiction.     
Now let $F(\widehat{R})=b$. Monotonicity of $F$ implies $F(R)\geq b$ for all $R$ such that $\widehat{R}\prec R$. Further, by monotonicity of $F$, $F(R)=a$, implies $R\prec \widehat{R}$. Thus $\widehat{R}$ is an upper bound of $S$. Therefore, $R^{0}$ is not the supremum of $S$. This is a contradiction.     
This implies neither Case $(i)$ nor Case $(ii)$ hold. Consequently $R^0=R_0$.

\noindent{\bf Step $3$:} Now we complete the proof of the claim.  
We show $F(R_0)=F(R^0)\in \{a,b\}$. 
If $F(R^{0})<a$, then $R^{0}\prec R'$ because $F(R')=a$ and $F$ is monotone. Thus $R^{0}$ is not the supremum of $S$.
If $F(R^{0})>b$, then $R''\prec R_{0}$ because $F(R'')=b$ and $F$ is monotone. Thus $R_{0}$ is not the infimum of $T$. Thus both cases lead to a contradiction. Thus we have $F(R_0)=F(R^0)\in \{a,b\}$.

\noindent We consider two cases. 
We have $a< b$, and by the contradiction hypothesis $aP''b$.

\noindent Case $(i):$ $F(R_0)=F(R^0)=a$. Since $F(R'')=b$, by monotonicity of $F$ it follows that $R_0\prec R''$. 
Hence by Lemma \ref{lemma:preference_preserve} $aP_0b$.
We will now show that this leads to a contradiction.
Consider a sequence $\{R^n\}_{n=1}^{\infty}$ such that for all $n$, $(i)$ $F(R^n)=b$ and $(ii)$ $R_0\prec R^{n+1}\prec R^{n}\prec R''$, $(iii)$ $R^n\rightarrow R_0$.
Since $\mathcal{R}^{cms}$ is metrizable, by Lemma \ref{lemma:closure} and richness of the domain such a sequence exists. Now $\lim_{n\rightarrow \infty}F(R^n)=b$. Thus by continuity of $V^F$, $b=\lim_{n\rightarrow \infty}F(R^n)I_0F(R_0)$. Since $F(R_0)=a$, and $aP_0b$ this is a contradiction. 

\medskip

\noindent Case $(ii):$ Suppose instead $F(R_0)=F(R^0)=b$. Since $R_{0}=\inf_{\prec}T$, $R_{0}\precsim R''$. 
Hence by Lemma \ref{lemma:preference_preserve} $aP_0b$. Also by monotonicity of $F$, $R'\prec R_0$
Consider a sequence $\{R^n\}_{n=1}^{\infty}$ such that for all $n$, $(i)$ $F(R^n)=a$ and $(ii)$ $R'\prec R^{n}\prec R^{n+1}\prec R_{0}$, $(iii)$ $ R^n\rightarrow R_0$.
Now $\lim_{n\rightarrow \infty}F(R_n)=a$. Again  by continuity of $V^F$, $a=\lim_{n\rightarrow \infty}F(R^n)I_0F(R_0)=b$. Since $F(R_0)=b$ and $aP_0b$, it contradicts the continuity of $V^{F}$.

\noindent An analogous argument leads to a contradiction if $bP'a$. Thus Claim \ref{claim:ab} is established. 

\begin{remark}\rm In the proof of Claim \ref{claim:ab} we have used preferences 
	in $[R',R'']$. Thus we can use this argument if the domain of $F$ were a closed interval 
	$[\underline{R},\overline{R}]$.  	     
\end{remark}

\begin{claim}\rm $aI_{0}b$. 
	\label{claim:indiff}
\end{claim}

\noindent{\bf Proof of Claim \ref{claim:indiff}:}  In Step $2$ in the proof of Claim \ref{claim:ab} we have obtained $R_{0}=R^{0}$. We establish that $aI_{0}b$. 
By the richness of $\mathcal{R}^{cms}$ there is $R$ such that $aIb$.

\begin{remark}\rm By Claim \ref{claim:ab} and Lemma \ref{lemma:preference_preserve}
	$R'\precsim R\precsim R''$. Thus we can use this argument if the domain of $F$ were a closed interval 
	$[\underline{R},\overline{R}]$.	  
\end{remark}

\noindent We show that $R=R_{0}=R^{0}$. First we show $F(R_{0})=F(R^{0})\in \{a,b\}$. If $F(R^{0})<a$, then $R^{0}$ is not the supremum of $S$ since $F(R')=a$ and thus $R^{0}\prec R'$ by monotonicity of $F$. If $F(R_{0})>b$, then $R_{0}$ is not the infimum of $T$ since $F(R'')=b$ and thus $R''\prec R_{0}$ by monotonicity of 
$F$. Therefore  $F(R^{0})<a$ and $F(R_{0})>b$ lead to a contradiction. Thus $F(R_{0})=F(R^{0})\in \{a,b\}$.

Let by way of contradiction $aP_{0}b$. Then by Claim \ref{claim:ab}, $F(R_{0})=a$. Since $aIb$, by Lemma \ref{lemma:preference_preserve} $R'\precsim R_{0}\prec R$. Since $R_{0}$ is the supremum of $S$, 
$F(R)\geq b$. By Claim \ref{claim:ab} $bR'' a$.
If $aI''b$, then by the single-crossing property $R=R''$. Therefore, by the single-crossing property or by Lemma \ref{lemma:preference_preserve} 
for all $R^*\in [R',R''[$, $aP^{*}b$. By Claim \ref{claim:ab} and monotonicity of $F$, $F(R^{*})=a$. Thus $R''=\inf_{\prec}T$. Since $aP_{0}b$, $R_{0}\neq R''$ and thus $R_{0}$ is not the infimum of $T$. This is a contradiction.     
Let $bP''a$. Then by the single crossing-property $R'\prec R\prec R''$. Since $F(R)\geq b$, by Claim \ref{claim:ab} and monotonicity of $F$, 
$F(R)=b$. By the single-crossing property for all $R^{*}\in [R',R[$, $aP^{*}b$. Thus by Claim \ref{claim:ab} and monotonicity of $F$, $F(R^{*})=a$. Since $aP_{0}b$, $R_{0}\neq R$ and thus $R_{0}$ is not the infimum of $T$. This is a contradiction.

Let by way of contradiction $bP_{0}a$.Then by Claim \ref{claim:ab}, $F(R_{0})=b$. Then $R\prec R_{0}\precsim R''$ where $aIb$. Since $R_{0}$ is the infimum of $T$, 
$F(R)\leq a$. By Claim \ref{claim:ab} $aR' b$. If $aI'b$, then by the single-crossing property $R=R'$. 
Then by the single-crossing property for all $R^{*}\in ]R',R'']$, $bP^*a$. By Claim \ref{claim:ab} and monotonicity of $F$, $F(R^{*})=b$. Thus $R'=\sup_{\prec} S$. Since $bR_{0}a$ and $R_{0}\neq R'$, $R_{0}$ is not the supremum of $S$. This is a contradiction. Let $aP'b$. Then by the single-crossing property 
$R'\prec R \prec R''$. Since $F(R)\leq a$ by Claim \ref{claim:ab} and monotonicity of $F$, $F(R)=a$.
Then, by the single-crossing property for all $R^{*}\in ]R,R'']$, $bP^{*}a$. Thus by Claim \ref{claim:ab} and monotonicity of $F$, $F(R^{*})=b$. Since $bP_{0}a$, $R_{0}\neq R$ and thus $R_{0}$ is not the supremum of $S$. This is a contradiction. 
Thus Claim \ref{claim:indiff} is established. 	

\begin{claim}\rm Let $F(R_{e})=e, F(R_{c})=c, F(R_{a})=a, F(R_{b})=b, F(R_{d})=d, F(R_{f})=f$. Then 
	$xR_{x}y$ for all $x,y\in Rn(F)$. 
	
	\label{claim:sp_all_bundles}	
\end{claim}

\noindent{\bf Proof of Claim \ref{claim:sp_all_bundles}:} By Claim \ref{claim:ab} 
$eR_{e}c$, $cR_ce,cR_ca, aR_ac,aR_ab, bR_ba,bR_bd, dR_da, dR_df, fR_fd$. 
Recall that $e<c<a<b<d<f$. Thus 
$F$ is locally strategy-proof in range. 
Now we show that $F$ is strategy-proof. From Claim \ref{claim:indiff} we have $R_{k},k=1,\ldots,5$ such that 
$eI_{1}c, F(R_1)\in \{e,c\}; cI_2a, F(R_{2})\in \{c,a\}; aI_{3}b, F(R_{3})\in \{a,b\}; bI_{4}d, F(R_{4})\in \{b,d\}; dI_{5}f, F(R_{5})\in \{d,f\}$. 
Further, $R_{1}\precsim R_{2}\precsim R_3 \precsim R_4  \precsim R_5$. 
Suppose by way of contradiction $R_{2}\prec R_1$. Since $eI_1c$, by the single-crossing property or by  Lemma \ref{lemma:preference_preserve} $ e P_{2} c$.
By Claim \ref{claim:ab}, $F(R_{2})\neq c$. That is if $F(R_2)=c$, then local strategy-proofness in range is violated. If $F(R_2)>c$, then monotonicity of $F$ is violated since $F(R_{1})\in \{e,c\}$.
Thus $R_{1}\precsim R_2$. Let $R_{3}\prec R_{2}$. Since $cI_2a$, by the single-crossing property 
$cP_3a$. By Claim \ref{claim:ab}, $F(R_{3})\neq a$. Thus if $F(R_3)=b$, then monotonicity of $F$ is violated because $F(R_2)\in \{c,a\}$. Thus $R_1\precsim R_2\precsim R_3$. By induction $R_1\precsim R_2\precsim R_3  \precsim R_4\precsim R_5$.

We show that $F$ restricted to $\{R_{i}\mid i=1,\ldots,5 \}$ is strategy-proof. 
We have $F(R_{1})\in \{e,c\}$ and $eI_1c$. Let $R_{1}\prec R_{2}$. 
We have $cI_{2}a$. By the single-crossing property or by Lemma \ref{lemma:preference_preserve} $cP_{1}a$.
Further we have $R_1\prec R_{3}$ and $aI_{3}b$. By the single-crossing property $aP_{1}b$. 
Thus by transitivity of $R_1$, $cP_{1}b$. 
In this way we obtain $eI_{1}cP_{1}x, x\notin \{e,c\}$.  
The arguments for $R_{2}$ are the same. 
Consider $R_{3}$. Let $R_{2}\prec R_{3}$. 
Since $cI_{2}a$, by the single-crossing property $aP_{3}c$.  
Since $eI_{1}c$, by the single-crossing property $cP_{3}e$. Thus by transitivity of $R_3$ $aP_{3}e$. The rest of the arguments are similar.   

Now we construct the strategy-proof mechanism.    
Consider $R\prec R_{1}$. By monotonicity of $F(R)\leq F(R_1)\leq c$.
By the single-crossing property $ePc$. 
Thus by Claim \ref{claim:ab}, $F(R)=e$.   
Since $R_{1}\precsim R_{2}$, it follows that $R\prec R_{2}$.
Since $cI_{2}a$ by the single-crossing property $cPa$. Thus by transitivity of $R$, $ePa$. Continuing this argument finitely many times we obtain $ePx$ for all $x\neq e$ and $x\in Rn(F)$. Thus for $R\prec R_{1}$, $F$ is strategy-proof.      

Now consider $R\in ] R_1,R_2[$.  By monotonicity of $F$, $F(R)\in \{e,c,a\}$. Since $eI_1c$, by the single-crossing property $cPe$. Since $cI_{2}a$, by the single-crossing property $cPa$. 
Thus by Claim \ref{claim:ab}, $F(R)=c$. Since $R\prec R_{3}$, and $aI_{3}b$, by the single-crossing property
$aPb$. Then $cPa$ implies $cPb$. This argument can be used finite number of times to show that $cPx$ for all $x \in Rn(F)$ such that $x\neq c$. Now it follows that $F$  restricted to $]\infty,R_{2}]$ is strategy-proof and the restriction is defined below,

$$F(R) = \begin{cases}
e, & \text{if $R \prec R_{1}$;}\\
\text{either}~ e~\text{or}~ c, & \text{if $R=R_1$}\\
c & \text{if $R_1\prec R \prec R_2$}\\
\text{either}~c~\text{or}~a & \text{if $R=R_2$.}
\end{cases}$$

\noindent For the sake of completion we extend the argument for $R\in ]R_2,R_3[$. Since $aI_{3}b$, by the single-crossing property for $R \prec R_3$, $aPb$. Since $cI_2a$ and
$R_{2}\prec R$, $aPc$. The continuation of this argument entails $aPx$ for all $x\neq a$ and $x\in Rn(F)$. Thus the general $F$ is defined as follows.

$$F(R)=\begin{cases}
e, & \text{if $R \prec R_{1}$;}\\
\text{either}~ e~\text{or}~ c, & \text{if $R=R_1$}\\
c & \text{if $R_1\prec R \prec R_2$}\\
\text{either}~c~\text{or}~a & \text{if $R=R_2$}\\
a & \text{if $R_2\prec R\prec R_3$}\\
\text{either}~a~\text{or}~b & \text{if $R=R_3$}\\
b & \text{if $R_3\prec R\prec R_4$}\\
\text{either}~b~\text{or}~d & \text{if $R=R_4$}\\
d & \text{if $R_4\prec R\prec R_5$}\\
\text{either}~d~\text{or}~f & \text{if $R=R_5$}\\
f & \text{if $R_5\prec R$.}
\end{cases}$$

\noindent To complete the proof we note $\#\{x\in Rn(F)\mid y\in Rn(F), xR_iy~ \text{for some}~R_{i}\in \{R_{1},\ldots, R_5 \}\}\leq 3$. This says that the special preferences $R_i$s, i.e., the preferences that are indifferent between two bundles in the range,   can be indifferent with at most one more bundle. To see this without loss of generality consider $R_{3}$, we have $aI_{3}b$. Let $a<x<b$, such that $aI_{3}bI_{3}x$ and $x\in Rn(F)$ . By the single-crossing property $aPx$ for $R\prec R_{3}$, and $bPx$ for $R_{3}\prec R$. 
By Claim \ref{claim:ab}, $F(R)\neq x$ for $R\neq R_{3}$. Thus $F(R_{3})=x$.
If there is any other $y\in Rn(F)$ with $a<y<b$, then $F(R_{3})=y$. Since $F$ is a function $x=y$.    

\begin{remark}\rm This constriction of $F$ holds for interval $[\underline{R},\overline{R}]$, since the construction is unaffected if there are no preferences $R$ such that $R\prec R_1$ or $R_{5}\prec R$.

\end{remark}

\noindent This completes a proof of the theorem. 

\medskip

\medskip

\noindent{\bf Proof of Lemma \ref{lemma:measurable}:} Let $Rn(F)=\{(t^i,q^i),\mid i=1,\ldots ,n\}$. For any $t^i$, 
$t^{-1}(\{t^{i}\})$ is an interval in $\mathcal{R}^{cms}$.  All intervals are in $\mathcal{B}$. Let $B\in B(\Re)$. Then $t^{-1}(B)=\cup_{j_i\in B}t^{-1}(\{t^{j_i}\})\in \mathcal{B}$. The proof for $q$ is analogous. Hence the lemma follows.

\medskip

\noindent{\bf Proof of Lemma \ref{lemma:soln_exists}:} In Step $1$ we show that the objective function is continuous, and then in Step $2$ we show that the constraint set is compact. Then from Step $1$ and $2$ the existence of maximum follows. For the sake of simplicity of notations, in this proof we denote preferences by $S$ and $ R$. 	  

\noindent{\bf Step $1$: The objective function is continuous} Consider the function: $M:[\underline{R},\overline{R}]\times [\underline{R},\overline{R}]\rightarrow [0,1] $ defined by $M(S,R)=\mu([S,R])$ if  $S\precsim R$ and $M(S,R)=\mu([R,S])$ if $R \prec S$. We show that this function is continuous, which follows because $\mu$ is continuous. To see this let $S^{n},R^{n}$ converge to $S,R$. 
Let without loss of generality $S\precsim R$, 
i.e., we have the interval $[S,R]$. Let $S\prec R$. Since $[\underline{R},\overline{R}]\times [\underline{R},\overline{R}]$ is metrizable, 
it is enough to show $\lim_{n\rightarrow \infty}M(S^{n},R^{n})=M(S,R)$, where $\lim_{n\rightarrow \infty}S^{n}=S, \lim_{n\rightarrow \infty}R^{n}=R$. 
Further, without loss of generality it is enough to consider monotone sequences. That is, $(a)$ $ S^{n} \uparrow S, R^{n}\uparrow R$,
$(b)$ $ S^{n} \downarrow S, R^{n}\uparrow R$ , $(c)$ $ S^{n} \downarrow S                                                                                                                                                                                                                                                                                                                                                                                                                                                                                                                                                                                                                                                                                                                                                                                                                                                                                                                                                                                                                                                                                                                                                                                                                                                                                                                                                                                                                   , R^{n}\downarrow R$, $(d)$ $ S^{n} \uparrow S, R^{n}\downarrow R$. To see this let $(S^n,R^n)\rightarrow (S,R)$ but
$\lim_{n\rightarrow \infty}M(S^n,R^n)\neq M(S,R)$. This means that there is an $\epsilon>0$ such that 
for all $N$ there is $n\geq N $ such that $|M(S^n,R^n)-M(S,R)|\geq \epsilon$. 
Let $\{(S^{n_k}, R^{n_k})\}_{k=1}^{\infty}$ be a subsequence such that $|M(S^{n_{k}},R^{n_{k}})-M(S,R)|\geq \epsilon$. 
Since $\lim_{k\rightarrow \infty }(S^{n_k}, R^{n_k})=(S,R)$, there is a further monotone subsequence of $(S^{n_k}, R^{n_k})$ that converges to $(S,R)$. But then along this monotone subsequence 
the values of $M$ will not converge to $M(S,R)$ which contradicts the assumption that along every monotone sequence converging to $(S,R)$ the values of the function $M$ converge.

Consider $(a)$ $S^{n} \uparrow S, R^{n}\uparrow R$. There exists $N$ such that $S\prec R^{n} $ for all $n\geq N$. Let $[S^{n},R^{n}]=
[S^{n},S]\cup]S,R^{n}]$. Then $[S^{n},S]\downarrow\{S\}$ and $ ]S,R^{n}]\uparrow]S,R]$. Since $\mu$ is a probability 
$\lim_{n\rightarrow \infty}\mu([S^{n},S])=\mu(\{S\})=0$ and 
$\lim_{n\rightarrow \infty} \mu(]S,R^{n}]=\mu]S,R]=\mu([S,R])$. Since $\mu$ is continuous, 
$\mu(]S,R])=\mu([S,R])$. The other cases can be proved analogously. The function $(t,S,R)\mapsto tM(S,R)$ is continuous; the objective function is the sum of $l$ such functions. Therefore, the objective function is continuous.

\noindent Next we show that the constraint set is compact.

\noindent{\bf Step $2$: The constraint set is compact} 
The constraint set can be rewritten as:  	

$C=\{(t^{0},t^{1},\ldots,t^{l-1}, q^{0},\ldots,q^{l-1}, R^{0},\ldots, R^{l})\mid
(t^{k},q^{k})I_{R^{k}}(t^{k-1},q^{k-1}),t^{k-1}\leq t^{k}, q^{k-1}\leq q^{k};
k=1\ldots,l-1, 
R^{k-1}\precsim R^{k},k=1,\ldots,l$, $R^{0}=\underline{R}$,  $R^{l}=\overline{R}\}$.

\noindent By individual rationality   
$F_l(\overline{R})\overline{R}(0,0)$. Let $(\overline{T},1)\overline{I}(0,0)$. 
By the single-crossing property for $R\prec \overline{R}$ if $(t,q)I(0,0)$, then 
$t<\overline{T}$. Thus by individual rationality let $F_l(R)\leq \overline{T}$ for all $F_l\in \mathcal{F}_l$.

\noindent Thus, 
\[C\subseteq [0,\overline{T}]\times [0,\overline{T}]\times\ldots \times[0,\overline{T}]\times\ldots\times[0,1]\times\ldots \times[0,1]\times \{\underline{R}\} \times [\underline{R},\overline{R}]\times\ldots [\underline{R},\overline{R}]\times \{\overline{R}\}\equiv \Sigma\]

\noindent Being a finite product of compact spaces, $\Sigma$ is compact in the product metric topology. We show that $C$ is closed in $\Sigma$. 
Let, \[x^{n}=\{(t^{0n},\ldots,t^{(l-1)n}, q^{0n},\ldots,q^{(l-1)l}, R^{0n}, R^{1n},\ldots, R^{ln})\}_{n=1}^{\infty}\]
\noindent  be a sequence in $C$ that converges to $x=(t^{0},\ldots,t^{(l-1)}, q^{0},\ldots,q^{(l-1)}, R^{0}, R^{1},\ldots, R^{l})$.
Inequalities are maintained in the limit. Further, by Lemma \ref{lemma:preference_conv} indifference is maintained in the limit. Thus, $x\in C$ and hence $C$ is a closed subset of a compact space. Thus $C$ is compact. 
This completes the proof of Lemma \ref{lemma:soln_exists}.

\end{appendix}

\end{document}